\documentclass[sigconf]{acmart}
\copyrightyear{2022}
\acmYear{2022}
\setcopyright{acmcopyright}\acmConference[WWW '22]{Proceedings of the ACM Web
Conference 2022}{April 25--29, 2022}{Virtual Event, Lyon, France}
\acmBooktitle{Proceedings of the ACM Web Conference 2022 (WWW '22), April
25--29, 2022, Virtual Event, Lyon, France}
\acmPrice{15.00}
\acmDOI{10.1145/3485447.3512155}
\acmISBN{978-1-4503-9096-5/22/04}

\usepackage{subfig}
\usepackage{algorithm}
\usepackage{algorithmic}
\usepackage{multirow}
\newcommand{\tabincell}[2]{\begin{tabular}{@{}#1@{}}#2\end{tabular}}
\AtBeginDocument{%
  \providecommand\BibTeX{{%
    \normalfont B\kern-0.5em{\scshape i\kern-0.25em b}\kern-0.8em\TeX}}}

\begin{document}






%

%
%

\newtheorem{problem}{Problem}
\newtheorem{assumption}{Assumption}
\newenvironment{sproof}{%
  \renewcommand{\proofname}{Sketch of Proof}\proof}{\endproof}
 \author{Junxiang Wang}
\affiliation{%
  \institution{Emory University}
  \streetaddress{201 Dowman Drive}
  \city{Atlanta}
  \state{Georgia}
  \country{United States}
  \postcode{30322}
}
\email{jwan936@emory.edu}

\author{Junji Jiang}
\affiliation{%
  \institution{Tianjin University}
  \streetaddress{92 Weijin Road}
  \city{Tianjin}
  \country{China}
  \postcode{22030}
}
\email{anjou_j@tju.edu.cn}

\author{Liang Zhao}
\affiliation{%
  \institution{Emory University}
  \streetaddress{201 Dowman Drive}
  \city{Atlanta}
  \state{Georgia}
  \country{United States}
  \postcode{22030}
}
\email{lzhao41@emory.edu}






\title{An Invertible Graph Diffusion Neural Network for Source Localization}

\begin{abstract}
Localizing the source of graph diffusion phenomena, such as misinformation propagation, is an important yet extremely challenging task in the real world. Existing source localization models typically are heavily dependent on the hand-crafted rules and only tailored for certain domain-specific applications. Unfortunately, a large portion of the graph diffusion process for many applications is still unknown to human beings so it is important to have expressive models for learning such underlying rules automatically. Recently, there is a surge of research body on expressive models such as Graph Neural Networks (GNNs) for automatically learning the underlying graph diffusion. However, source localization is instead the inverse of graph diffusion, which is a typical inverse problem in graphs that is well-known to be ill-posed because there can be multiple solutions and hence different from the traditional (semi-)supervised learning settings. This paper aims to establish a generic framework of invertible graph diffusion models for source localization on graphs, namely Invertible Validity-aware Graph Diffusion (IVGD), to handle major challenges including 1) Difficulty to leverage knowledge in graph diffusion models for modeling their inverse processes in an end-to-end fashion, 2) Difficulty to ensure the validity of the inferred sources, and 3) Efficiency and scalability in source inference. Specifically, first, to inversely infer sources of graph diffusion, we propose a graph residual scenario to make existing graph diffusion models invertible with theoretical guarantees; second, we develop a novel error compensation mechanism that learns to offset the errors of the inferred sources. Finally, to ensure the validity of the inferred sources, a new set of validity-aware layers have been devised to project inferred sources to feasible regions by flexibly encoding constraints with unrolled optimization techniques. A linearization technique is proposed to strengthen the efficiency of our proposed layers. The convergence of the proposed IVGD is proven theoretically. Extensive experiments on nine real-world datasets demonstrate that our proposed IVGD outperforms state-of-the-art comparison methods significantly. We have released our code at \url{ https://github.com/xianggebenben/IVGD}.
\end{abstract}
\maketitle
\copyrightyear{2022}
\acmYear{2022}
\setcopyright{acmcopyright}\acmConference[WWW '22]{Proceedings of the ACM Web
Conference 2022}{April 25--29, 2022}{Virtual Event, Lyon, France}
\acmBooktitle{Proceedings of the ACM Web Conference 2022 (WWW '22), April
25--29, 2022, Virtual Event, Lyon, France}
\acmPrice{15.00}
\acmDOI{10.1145/3485447.3512155}
\acmISBN{978-1-4503-9096-5/22/04}
\section{Introduction}
\indent Graphs are prevalent data structures where nodes are connected by their relations. They have been widely applied in various domains such as social networks \cite{scott1988social}, biological networks \cite{junker2011analysis}, and information networks \cite{harvey2006capacity}. As a fundamental task in graph mining, graph diffusion aims to predict future graph cascade patterns given source nodes. However, its inverse problem, graph source localization, is rarely explored and yet is an extremely important topic. It aims to detect source nodes given their future graph cascade patterns. As an example shown in Figure \ref{fig:illustration}, the goal of graph diffusion is to predict the cascade pattern $\{b,c,d,e\}$ given a source node $b$; while the goal of graph source localization is to detect source nodes $b$ or $c$ given the cascade pattern $\{b,c,d,e\}$. Graph source localization covers a wide range of promising research and real-world applications. For example, misinformation such as ``drinking bleach or alcohol can prevent or kill the virus" \cite{islam2020covid} in social networks is required to detect as early as possible, in order to prevent it from spreading;  Email is a primary vehicle to transmit computer viruses, and thus tracking the source Emails carrying viruses in the Email networks is integral to computer security \cite{newman2002email}; malware detection aims to position the source of malware in the Internet of Things (IoT) network \cite{jang2014mal}.  Therefore, the graph source localization problem entails attention and extensive investigations from machine learning researchers.\\
\indent The forward process in Figure \ref{fig:illustration}, namely, graph diffusion, has been studied for a long time, by traditional prescribed methods based on hand-crafted rules and heuristics such as SEHP \cite{bao2015modeling}, OSLOR \cite{cui2013cascading}, and DSHP \cite{ding2015video}. Following similar styles of traditional graph diffusion methods, classical methods for its inverse process, namely source localization of graph diffusion, have also been dominated by prescribed approaches. Specifically, a majority of methods are based on predefined rules, by utilizing either heuristics or metrics to select sources such as distance errors. Some other prescribed methods partition nodes into different clusters based on network topologies, and select source nodes in each cluster. These prescribed methods rely heavily on human predefined heuristics and rules and usually are specialized for specific applications. Therefore, they may not be suitable for applications where prior knowledge on diffusion mechanisms is unavailable.  Recently, with the development of  GNNs \cite{GNNBook2022}, Dong et al. utilized state-of-the-art architectures such as Graph Convolutional Network (GCN) to localize the source of misinformation \cite{dong2019multiple}. However, their method requires results from prescribed methods as its input, and hence still suffers from the drawback of prescribed methods mentioned above.\\
\begin{figure}
    \centering
    \includegraphics[width=\linewidth]{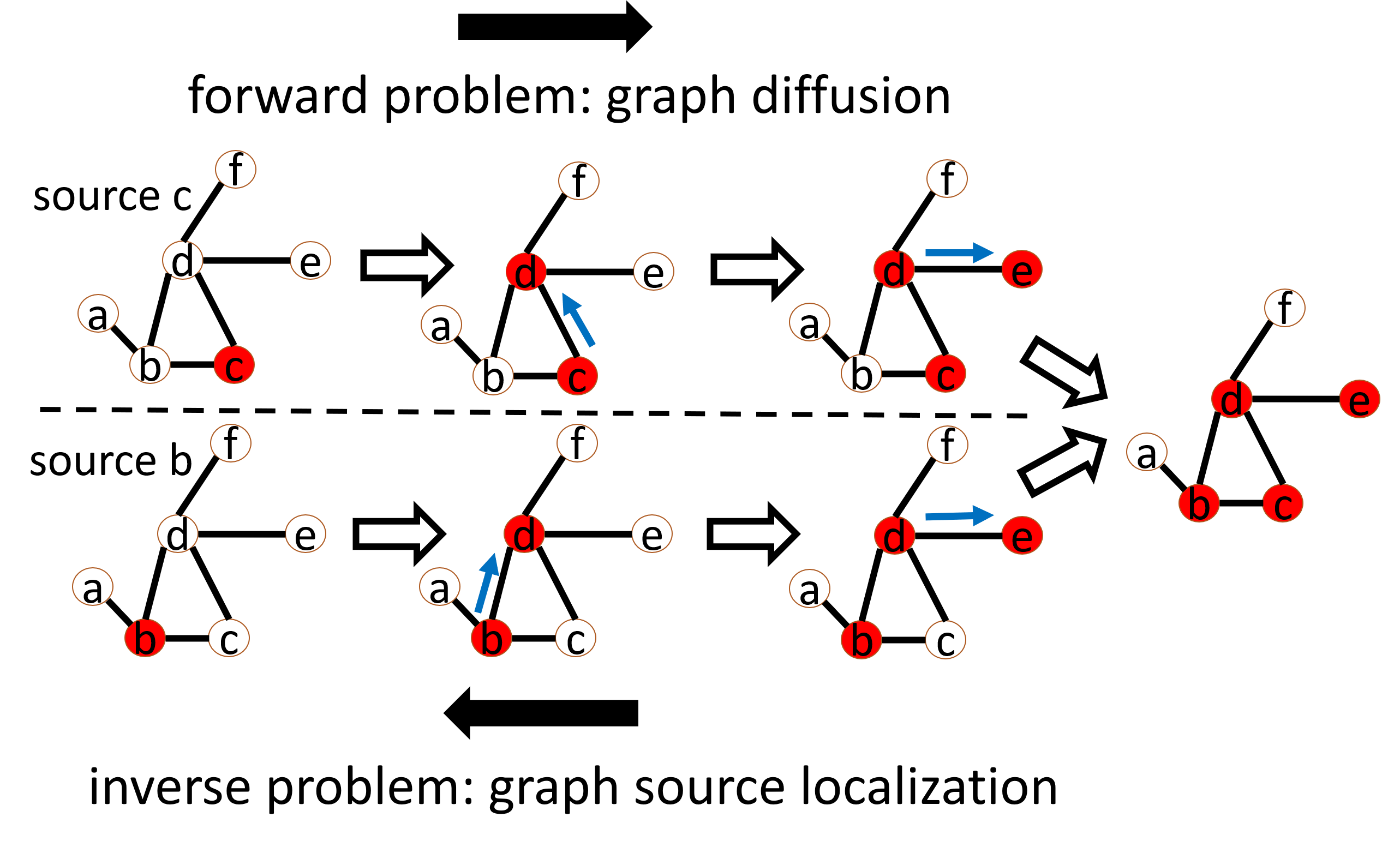}
    \vspace{-0.3cm}
    \caption{An example of information diffusion: different source nodes generate the same cascade pattern.}
    \label{fig:illustration}
    \vspace{-0.5cm}
\end{figure}
\indent In recent years, the advancement of GNNs leads to state-of-the-art performance in many graph mining tasks such as node classification and link prediction. They can incorporate node attributes into models and learn node representations effectively by capturing network topology and neighboring information  \cite{kipf2016semi}. They have recently expanded their success into graph diffusion problems \cite{cao2020popularity}, by tackling the drawbacks of traditional prescribed methods in graph diffusion. specifically, instead of requiring prior knowledge and rules of diffusion, GNNs based methods can "learn"  rules from the data in an end-to-end fashion. Although GNNs have been well applied for performing graph diffusion tasks, however, it is difficult to devise their inverse counterparts (i.e. graph source localization models) because such an inverse problem is much more difficult and involves three key challenges:   \textbf{1). Difficulty to leverage knowledge in graph diffusion models for modeling their inverse processes in an end-to-end fashion.}  The learned knowledge from graph diffusion models facilitates source localization. For example, as shown in Figure \ref{fig:illustration}, while nodes $b$ and $c$ generate the same cascade pattern  $\{ b,c,d,e\}$,  the learned knowledge from graph diffusion models is useful to predict which node is likely to be the source. However, it is extremely challenging to incorporate such a notion into the inverse problem in an end-to-end manner, and it is prohibitively difficult to define hand-crafted ways to achieve it with graph diffusion models directly since they are opposite processes. \textbf{2). Difficulty to ensure the validity of the inferred sources.} Graph sources usually follow validate graph patterns. For example, in the application of misinformation detection, sources of misinformation should be connected in the social networks. As another example, sources of malware are dense in some restricted regions of the IoT networks. Such validity constraints are imposed both in the training and test phases, which should be achieved by delicately-designed activation layers. Traditional activation layers such as softmax are exerted on individual nodes. However, validity constraints require the projection of multiple sources by considering their topological connections. \textbf{3). Efficiency and scalability in source inference.} Inferring sources constrained by validity patterns is a combinatorial problem and hence is time-consuming. To multiply the difficulty, the inverse process of graph diffusion models should also be inferred. Therefore, devising a scalable and efficient algorithm is important yet challenging.\\
\indent  In this paper, we propose a novel Invertible Validity-aware Graph Diffusion (IVGD) to simultaneously tackle all these challenges. Specifically, given a graph diffusion model, we make it invertible by restricting its Lipschitz constant for the residual GNNs, and thus an approximate estimation of source localization can be obtained by its inversion, and then a compensation module is presented to reduce the introduced errors with skip connection. Moreover, we leverage the unrolled optimization technique to integrate validity constraints into the model, where each layer is encoded by a constrained optimization problem. To combat efficiency and scalability problems, a linearization technique is used to transform problems into solvable ones, which can be efficiently solved by closed-form solutions. Finally, the convergence of the proposed IVGD to a feasible solution is proven theoretically. Our contributions in this work can be summarized as follows:
\begin{itemize}
    \item \textbf{Design a generic end-to-end framework for source location.} We develop a framework for the inverse of graph diffusion models, and learn rules of graph diffusion models automatically. It does not require hand-crafted rules and can be used for source localization. Our framework is generic to any graph diffusion model, and the code has been released publicly.
    \item \textbf{Develop an invertible graph diffusion model with an error compensation mechanism.} We propose a new graph residual net with Lipschitz regularization to ensure the invertibility of graph diffusion models. Furthermore, we propose an error compensation mechanism to offset the errors inferred from the graph residual net. 
    \item \textbf{Propose an efficient validity-aware layer to maintain the validity of inferred sources.} Our model can ensure the validity of inferred sources by automatically learning validity-aware layers. We further accelerate the optimization of the proposed layers by leveraging a linearization technique. It transforms nonconvex problems into convex problems, which have closed-form solutions. Moreover, we provide the convergence guarantees of the proposed IVGD to a feasible solution.
    \item \textbf{Conduct extensive experiments on nine datasets.}  Extensive experiments on nine datasets have been conducted to demonstrate the effectiveness and robustness of our proposed IVGD. Our proposed IVGD outperforms all comparison methods significantly on five metrics, especially $20\%$ on F1-Score.
\end{itemize}
\section{Related Work}
\label{sec:related work}
\indent  In this work, we summarize existing works related to this paper, which are shown as follows:\\
\indent\textbf{Graph Diffusion:} Graph diffusion is the task of predicting the diffusion of information dissemination in networks. It has a wide range of real-world applications such as societal event prediction \cite{wang2018incomplete,zhao2021event,zhao2017spatial}, and adverse event detection in social media \cite{wang2018multi,wang2018semi}. A large number of research works have been conducted to improve the quality of predictions. Most existing works usually assume the topologies of networks and apply the classical probabilistic graphical models. For example, Ahmed et al. identified patterns of temporal evolution that are generalizable to distinct types of data \cite{ahmed2013peek}; Bandari et al. constructed a multi-dimensional feature space derived from properties of an article and evaluate the efficacy of these features to serve as predictors of online diffusion \cite{bao2015modeling}. More traditional methods can be found in various survey papers \cite{moniz2019review,gao2019taxonomy,zhou2021survey}. However, they are applicable to a specific type of neural network and are poorly generalizable. A recent line of research works use Recurrent Neural Networks (RNN) to predict the diffusion, and usually include multimodality such as text content and time series \cite{bielski2018understanding,chen2019npp,chu2018cease,xie2020multimodal,zhang2018user}. Various techniques have been applied including self-attention mechanism \cite{bielski2018understanding,chen2017attention,wang2017cascade}, knowledge base \cite{zhao2019neural}, multi-task learning \cite{chen2019information}, and stochastic processes \cite{cao2017deephawkes,du2016recurrent,liao2019popularity} . However, they cannot utilize network topology to enhance predictions. To handle this challenge, GNNs have been applied to predict either macro-level (i.e. global level) tasks \cite{cao2020popularity} or micro-level (i.e. node level) tasks \cite{he2020network,qiu2018deepinf,wang2017topological,xia2021DeepIS} combined with RNN, and a handful of works attempted to utilize other neural network architectures \cite{kefato2018cas2vec,sanjo2017recipe,wang2018factorization,yang2019neural}.\\
\indent\textbf{Graph Source Localization:} The goal of the graph source localization is to identify the source of a network based on observations such as the states of the nodes and a subset of timestamps at which the diffusion process reached the corresponding nodes \cite{ying2018diffusion}. Graph source localization has a wide range of applications such as disease localization, virus localization, and rumor detection. Several recent surveys on this topic are available \cite{shah2010detecting,jiang2016identifying,shelke2019source}. Similar to graph diffusion models, existing graph source localization papers usually require the assumptions of the diffusion, network topology, and observations. With the development of GNNs, Dong et al. proposed a Graph Convolutional Networks based Source Identification (GCNSI) model for multiple source localization \cite{dong2019multiple}. However, its model relies heavily on hand-crafted rules. Moreover, its performance suffers from the class imbalance problem, as shown in experiments.
\begin{table}[]
\tiny
    \centering
    \begin{tabular}{c|c}
    \hline\hline
         Notations&Descriptions\\
         \hline
         $V$& Node set \\
         \hline
         $E$& Edge set \\
         \hline
         $Y_t$& Diffusion vector at time $t$\\ \hline
         $x$ & The vector of source nodes \\ \hline
         $f_W$& The function of feature construction \\ \hline
         $g$& The function of label propagation\\ \hline
         $\Phi(x)=0$ & The equality constraint of a validity pattern\\ \hline
         $T$ & The length of diffusion \\ \hline\hline
    \end{tabular}
    \caption{Important notations and descriptions}
    \label{tab:notation}
    \vspace{-0.8cm}
\end{table}
\section{Problem Setup}
\indent In this section, the problem addressed by this research is formulated mathematically in the form of an inverse problem.
\subsection{Problem Formulation}
\label{sec:problem formulation}
\indent Important notations are outlined in Table \ref{tab:notation}. Consider a graph $G=(V,E)$, where $V=\{v_1,\cdots,v_n\}$ and $E$ are the node set and the edge set respectively, $\vert V\vert=n$ is the number of nodes. $Y_t\in \{0,1\}^{n}$ is a diffusion vector at time $t$. $Y_{t,i}=1$ means that node $i$ is diffused, while $Y_{t,i}=0$ means that node $i$ is not diffused.   $S$ is a set of source nodes. $x\in \{0,1\}^n$ is a vector of source nodes, $x_i=1$ if $v_i\in S$ and $x_i=0$ otherwise. The diffusion process begins at timestamp 0 and terminates at timestamp $T$. While there are many existing GNN-based graph diffusion models, a general GNN framework consists of two stages: feature construction and label propagation. In the feature construction, a neural network $f_W$ is learned to estimate the initial node diffusion vector $\zeta=f_{W}(x)$ based on input $x$, where $W$ is a set of  learnable weights in $f_W$. In the label propagation, a propagation function $g$ is designed to diffuse information to neighboring nodes: $Y_T=g(\zeta)$. Therefore, the graph diffusion model is $\theta=g(f_W(x))$, and its inverse problem, graph source localization, is to infer $x$ from $Y_{T}$. Moreover, a validity pattern can be imposed on sources in the form of the constraint $\Phi(x)=0$ such as the number of source nodes, and the connectivity among multiple sources. Then the graph source localization problem can be mathematically formulated as follows:
\begin{align}
    \theta^{-1}: Y_T \rightarrow x \quad s.t. \ \Phi(x)=0. \label{eq:source localization}
\end{align}
\subsection{Challenges}
\label{sec:challenge}
\indent It is extremely challenging to automatically learn the source localization model $\theta^{-1}$ and solve the problem in Equation \eqref{eq:source localization} given an arbitrarily complex forward model such as deep neural networks due to several key challenges: \textit{1). The difficulty to integrate information from $\theta$ into $\theta^{-1}$}. The complex graph diffusion model $\theta$ is typically not invertible directly, so it is challenging to transfer the knowledge from $\theta$ into its graph inverse problem. \textit{2). The difficulty to incorporate $\Phi(x)=0$ into $\theta^{-1}$.} $\Phi(x)=0$ considers topological connections of all nodes instead of an individual node, and it can express complex validity patterns because of its nonlinearity. So it is difficult to encode such validity information from all nodes to activation layers. \textit{3). Efficiency and scalability to solve Equation \eqref{eq:source localization}}. Solving Equation \eqref{eq:source localization} is a combinatorial problem because $Y_t$ and $x$ are discrete. So it is imperative to develop an algorithm to solve it efficiently, and to scale well on large-scale graphs (i.e. $n$ is very large).
\section{Proposed IVGD Framework}
\label{sec:proposed model}
\indent In this section, we propose a generic framework for graph source localization, namely Invertible Validity-aware Graph Diffusion (IVGD) to address these challenges simultaneously. The high-level overview of the proposed IVGD framework is highlighted in Figure \ref{fig:framework}. Specifically, our proposed IVGD consists of two components: in Figure \ref{fig:framework}(a), we propose an invertible graph residual net to address Challenge 1, where the approximate estimation of graph source localization can be obtained by inverting the graph residual net with the integration of the proposed error compensation module (Section \ref{sec:invertible graph residual net}); in Figure \ref{fig:framework}(b), a series of validity-aware layers are introduced to resolve Challenges 2 and 3, which encode validity constraints into problems with unrolled optimization techniques. With the introduction of the linearization technique, they can be solved efficiently with closed-form solutions (Section \ref{sec:validity layers}). We provide the convergence guarantees of our proposed IVGD to a feasible solution (Section \ref{sec:convergence}).

\subsection{The Invertible Graph Residual Net}
\label{sec:invertible graph residual net}
\indent  Our goal in this subsection is to obtain an approximate estimation of the source vector $x$ based on the learned knowledge from the graph diffusion model $\theta$. One intuitive idea is to invert the process of the forward model $\theta$. The key challenge here is that $\theta$ is not necessarily invertible, so the task is that how to devise an invertible architecture based on $\theta$. To address this, we propose a novel invertible graph residual net and provide theoretical guarantees to ensure invertibility. After an approximate estimation of the source vector $z$ is obtained by the proposed invertible graph residual net, a simple compensation module is introduced to reduce estimation errors, which is denoted as $x=C(z)$. Because $z$ is close to $x$, we utilize an MLP module $Q$ to measure the deviation of $z$ from $x$: $z^{'}=Q(z)$. A skip connection concatenates $z$ and $z^{'}$ to form the compensated prediction $z^{''}=z+z^{'}=z+Q(z)$. However, $z^{''}$ may be beyond range (i.e. smaller than 0 or larger than 1). In order to remove  such bias, a piecewise-linear function is utilized to truncate bias as follows: $ x=\min(\max(0,z^{''}),1)$.\\
\indent Now we aim to devise an invertible GNN-based architecture. While there are many classic invertible architectures such as i-Revnet and Glow \cite{jacobsen2018revnet,chang2018reversible,kingma2018glow}, their forms are quite complex and require extra components to ensure one-to-one mapping. i-ResNet, However, stands out among others because of its simplicity and outstanding performance, and it allows for the form-free design of layers \cite{behrmann2019invertible}. We extend the idea of the i-Resnet to the GNN by regularizing its Lipschitz coefficient. To achieve this, we first formulate the graph residual net of the general GNN framework. $\zeta =F_W(x)=(f_W(x)+x)/2$ and $Y_T=G(\zeta)=(g(\zeta)+\zeta)/2$ are graph residual blocks of feature construction and label propagation, respectively. $P(x)=G(F_W(x))$ denotes the graph residual net for graph diffusion, and $P^{-1}$ denotes its inverse  for graph source localization. Next, $P$ can be inversed to $P^{-1}$ by simply fixed point iterations. Algorithm \ref{algo:backward invertible graph residual net} demonstrates the inverse process of the graph residual net for source localization. Specifically, Line 1 and Line 5 are initializations of label propagation and feature construction, respectively. Lines 2-4 and Lines 6-8 are fixed-point iterations of label propagation and feature construction, respectively.
\begin{figure}
    \centering
    \includegraphics[width=\linewidth]{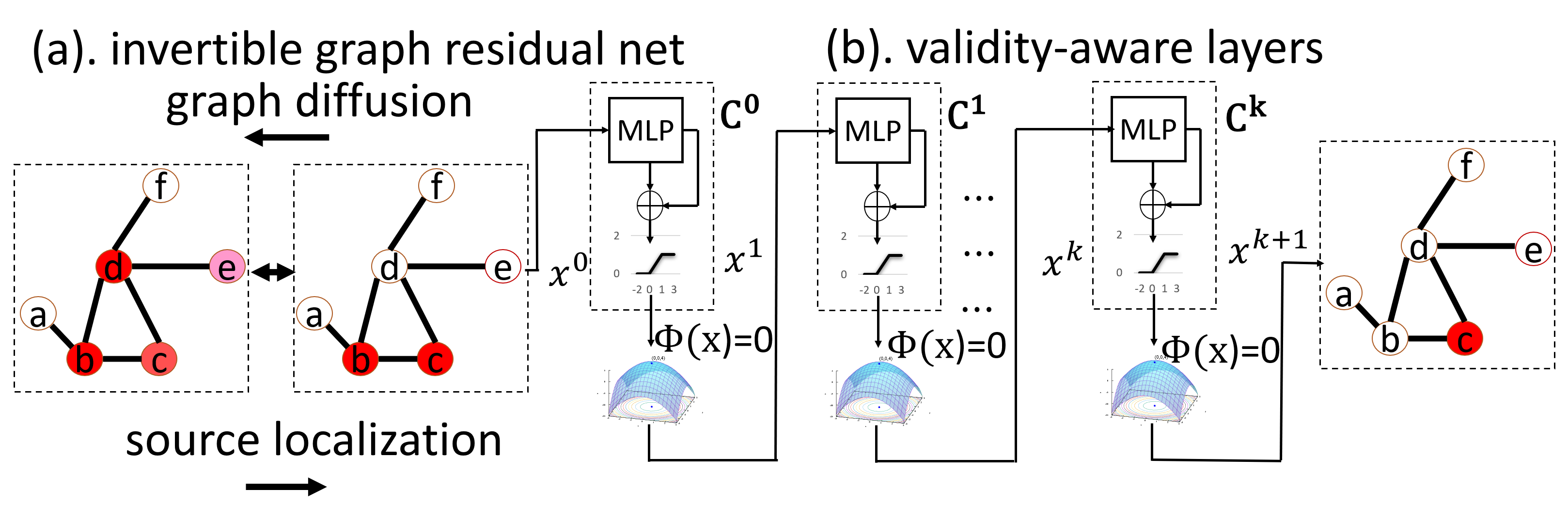}
    \caption{Framework overview: the proposed IVGD framework consists of an invertible graph diffusion model and a series of validity-aware layers.}
    \label{fig:framework}
    \vspace{-0.3cm}
\end{figure}

\begin{algorithm} 
\small
\caption{Inverse of the Graph Residual Net for Graph Source Localization} 
\begin{algorithmic}[1] 

\REQUIRE $f_W, g, Y_T, m$ \COMMENT{the number of iterations}. 
\ENSURE $z$ \COMMENT{an estimation of the source vector $x$}. 
\STATE $\zeta^0=Y_T$ \COMMENT{Initialization of label propagation}
\FOR{i=1 to m} 
\STATE $\zeta^i=2Y_T-g(\zeta^{i-1})$ \COMMENT{Fix point iteration of label propagation}
\ENDFOR
\STATE $z^0=\zeta^m$ \COMMENT{Initialization of feature construction}
\FOR{i=1 to m} 
\STATE $z^i=2\zeta^m-f_W(z^{i-1})$ \COMMENT{Fix point iteration of feature construction}
\ENDFOR
\STATE Output $z$.
\end{algorithmic}
\label{algo:backward invertible graph residual net}
\end{algorithm}

Next, we provide theoretical guarantees on the invertibility of the graph residual net. Specifically, we prove a sufficient condition to ensure invertibility and discuss practical issues to satisfy such conditions. The following theorem provides a sufficient condition for the invertibility of the graph residual net.
\begin{theorem}[Sufficient Condition for the invertibility of the graph residual net]
The graph residual net $P$ is invertible if $L_f<1$ and $L_g<1$, where $L_f$ and $L_g$ are Lipschitz constants of $f_W(x)$ and $g(\zeta)$, respectively. 
\label{theo:sufficient conditions}
\end{theorem}
\begin{proof}
Because $P=G(F_W)$, $P$ is invertible if $F_W$ and $G$ are invertible. We have $x=2\zeta-f_W(x)$ and $\zeta=2Y_T-g(\zeta)$ by the definitions of $F_W$ and $G$, and rewrite them as iterations as follows:
\begin{align*}
    &x^0=\zeta \ \text{and} \ x^{k+1}=2\zeta-f_W(x^k)\\& \zeta^0=Y_T \ \text{and} \ \zeta^{k+1}=2Y_T-g(\zeta^k) 
\end{align*}
where $\lim\nolimits_{k\rightarrow\infty}x^k=x$ and $\lim\nolimits_{k\rightarrow\infty}\zeta^k=\zeta$ are fixed points if $x^k$ and $\zeta^k$ converge. Because $f_W$ and $g$ are operators on a Banach space, $L_f<1$ and $L_g<1$ guarantee convergence by the Banach fixed point theorem \cite{behrmann2019invertible}.
\end{proof}
Then the following lemma provides the upper bound of the Lipschitz constants of the graph residual net.
\begin{lemma}[The Lipschitz constants of the graph residual net]
Let $L_P$ and $L_{P^{-1}}$ be Lipschitz constants of $P$ and $P^{-1}$, respectively, then $L_P\leq\frac{(1+L_f)(1+L_g)}{4}$, and $L_{P^{-1}}\leq \frac{4}{(1-L_f)(1-L_g)}$.
\label{lemma:lipschitz constant}
\end{lemma}
\begin{sproof}
To prove this lemma, we need to show that for any $x^{'},x^{''},\Vert P(x^{''})-P(x^{'})\Vert\leq \frac{(L_f+1)(L_g+1)}{4}\Vert x^{''}-x^{'}\Vert$, and for any  $y^{'},y^{''},
    \Vert P^{-1}(y^{''})-P^{-1}(y^{'})\Vert\leq \frac{4}{(1-L_f)(1-L_g)}\Vert y^{''}-y^{'}\Vert$. The complete proof is shown in Section \ref{sec:proof lipschitz constant} in the appendix.
\end{sproof}
Based on Theorem \ref{theo:sufficient conditions} and Lemma \ref{lemma:lipschitz constant}, we conclude that the Lipschitz constant of $P$  is less than 1 and therefore $P$ is invertible. Now we briefly discuss how to guarantee the Lipschitz constraints in practice. For $F_W$, which contains a set of learnable weights $W$, the power iteration method can be applied to normalize $W$ so that its norm is smaller than 1 \cite{gouk2021regularisation}. For $G$, many classical propagation functions such as the Independent Cascade (IC) function satisfies this condition \cite{xia2021DeepIS}.

\subsection{Validity-aware Layers}
\label{sec:validity layers}
\indent   The invertible graph residual net provides an estimation of source localization, However, the validity constraint $\Phi(x) = 0$ is still required to satisfy. Specifically, we aim to resolve the following optimization problem:
\begin{align}
    \min\nolimits_{x} R(x) \quad s.t. \quad x=C(z), \quad \Phi(x)= 0. \label{eq:original problem}
\end{align}
where  $R(x)$ is a loss function, and $z$ is the output of the graph residual net in Algorithm \ref{algo:backward invertible graph residual net}. However, Equation \eqref{eq:original problem} is unsolvable due to the potential constraint conflict $x=C(z)$ and $\Phi(x)= 0$. To address this, Equation \eqref{eq:original problem} is reduced to the following optimization problem:
\begin{align}
     \min\nolimits_{x} R(x)+\frac{\tau}{2}\Vert x-C(z)\Vert^2_2 \label{eq:unroll} \ \ s.t. \ \Phi(x)= 0. 
\end{align}
where $\tau>0$ is a tuning parameter to balance the loss and the error compensation module. Then the task here is to design activation layers, in order to solve Equation \eqref{eq:unroll}. While traditional activation layers focus on individual nodes, and cannot handle difficult constraints, unrolled optimization techniques are potential ways to incorporate complex validity patterns into the model. Motivated by the recent development of ADMM-Net \cite{sun2016deep} and OptNet \cite{Amos2017OptNet}, a potential solution to Equation \eqref{eq:unroll} can be achieved by unrolling the problem into a neural net, where each layer is designed for the following optimization problem: 
\begin{align}
x^{k+1}&\leftarrow \arg\min\nolimits_{x} J^k(x) \ s.t. \ \Phi(x)= 0, \label{eq:layer problem}
\end{align}
where $J^k(x)=R(x)+\frac{\tau^k}{2}\Vert x-C^k(x^k)\Vert^2_2$ and $x^0=z$. $x^k$ and $x^{k+1}$ are the input and the output of the $k$-th layer, respectively. To solve Equation \eqref{eq:layer problem},
the augmented Lagrangian function is formulated mathematically as follows \cite{bertsekas2014constrained}:
\begin{align*}
    H^k(x,\lambda)=J^k(x)+\Psi^k(x,\lambda),
\end{align*}
where $\Psi^k(x,\lambda)=\frac{1}{2\rho^k}((\lambda+\rho^k\Phi(x))^2-\lambda^2)$, $\rho>0$ is a hyperparameter, and $\lambda$ is a dual variable to address $\Phi(x)=0$. To optimize $H^k(x,\lambda)$, the OptNet updates variables via the implicit gradients of the Karush–Kuhn–Tucker (KKT) conditions \cite{Amos2017OptNet}. However, its computational efficiency is limited  due to the nonconvexity and nonlinearity of $\Phi(x)$, and scales poorly on the large-scale networks (i.e. Challenge 3 in Section \ref{sec:challenge}). To address this, we utilize a linearization technique to transform the nonconvex $H^k(x,\lambda)$ to the convex $h^k(x)$ as follows \cite{yang2013linearized}:
\begin{align*}
 h^k(x)&=J^k(x)+\!\partial_x (\Psi^k)^T(x^k,\lambda^k)(x-x^k)+\frac{\alpha^k}{2}\Vert x-x^k\Vert^2_2 ,
\end{align*}
where $\alpha^k>0$ is a hyperparameter to control the quadratic term.
 We formulate the validity-aware layer as follows:
\begin{align*}
    x^{k+1}&\leftarrow\arg\min\nolimits_{x}  h^k(x),\\
    \lambda^{k+1}&\leftarrow \lambda^k+\rho^k \Phi(x^{k+1}).
    \end{align*}
    Specifically, $C^k,\rho^k, \tau^k, \alpha^k$ can be considered as learnable parameters of the $k$-th layer. Notice that if $R(x)$ is a mean square error, then $h^k(x)$ is quadratic and has a closed-form solution. Validity-aware layers can be trained by state-of-the-art optimizers such as SGD and Adam \cite{Diederik2015Adam}.
\subsection{Convergence of the Proposed IVGD}
\label{sec:convergence}
\indent Like other unrolled optimization models which solve objective functions effectively, our proposed IVGD can address Equation \eqref{eq:unroll} by closed-form solutions. However, there lacks an understanding of the convergence of unrolled optimization models. This is because they usually involve many learnable parameters, which complicate the investigation of convergence. In this section, we provide the convergence guarantees of the proposed IVGD for the linear constraint $\Phi(x)=Ax-b$ where $A$ and $b$ are a given matrix and vector, respectively. Specifically, we propose a novel condition based on learnable parameters to ensure $x^{k}$ and $\lambda^{k}$ are closer to a solution as layers go deeper. Due to the space limit, we show the sketches of all proofs, and the complete proofs are available in the appendix.\\
\indent The optimality conditions of Equation \eqref{eq:layer problem} are shown as follows:
  \begin{align*}
        Ax^k_*-b=0, \ \nabla J^k(x^k_*)+A^T\lambda^k_*=0,
    \end{align*}
    where $(x^k_*,\lambda^k_*)$ is an optimal solution (not necessarily unique) to Equation \eqref{eq:layer problem}, which depends on $\tau^k$ and $C^k$.  The following lemma provides the relationship between $(x^k_*,\lambda^k_*)$ and $(x^{k+1},\lambda^{k+1})$.
    \begin{lemma}
For any $k\in \mathbb{N}$, it satisfies
\begin{align*}
   &\frac{1}{\rho^k}(\lambda^{k}-\lambda^{k+1})^T(\lambda^{k+1}-\lambda^k_*) +\alpha^k(x^k-x^{k+1})^T(x^{k+1}-x^k_*)\\&\geq (x^{k+1}-x^k)^TA^T(\lambda^k-\lambda^{k+1}).
\end{align*} \label{lemma: preliminary}
\end{lemma}
\begin{sproof}
It can be obtained by the optimality conditions of $x^{k+1}$ and $x^k_*$. The complete proof is shown in Section \ref{sec:proof preliminary} in the appendix.
\end{sproof}
\indent Motivated by Lemma \ref{lemma: preliminary}, we 
    let $u_1=(x_1,\lambda_1)$ and $u_2=(x_2,\lambda_2)$, and define an inner product by
    \begin{align}
\langle  u_1,u_2\rangle_{M^k}=\frac{1}{\rho^k} \lambda_1^T \lambda_2+\alpha^k x^T_1 x_2. \label{eq: notation}
    \end{align}
     and the induced norm $\Vert u\Vert^2_{M^k}=\langle u,u\rangle_{M^k}$. Denote $u^k_*=(x^k_*,\lambda^k_*)$ and $u^k=(x^k,\lambda^k)$. The following theorem states that $u^{k+1}$ is a feasible solution to Equation \eqref{eq:layer problem}.
    \begin{theorem}[Asymptotic Convergence]
     Assume $0<D_1\leq \alpha^k\leq D_2<\infty$, $0<D_3\leq \rho^k\leq D_4<\infty$  and $\alpha^k-\rho^kr(A^TA)>0$ where $D_1$, $D_2$, $D_3$ and $D_4$ are constant, and $r(A^TA)$ denotes the spectral radius of $A^TA$. If there exists $(C^k,\rho^k, \tau^k,\alpha^k)$ such that $\Vert u^{k+1}-u^{k+1}_{*}\Vert^2_{M^{k+1}}\leq \Vert u^{k+1}-u^{k}_{*}\Vert^2_{M^{k}}$, then we have \\(a). $\Vert u^k-u^{k+1}\Vert^2_{M^k}\rightarrow 0$.\\ (b). $\Vert u^k-u^{k}_*\Vert^2_{M^k}$ is nonincreasing and hence converges.\\
    (c). $u^{k+1}$ is a feasible solution to Equation \eqref{eq:layer problem} . That is,\\ $\lim\nolimits_{k\rightarrow\infty} Ax^{k+1}\!-\!b=0$, $\lim\nolimits _{k\rightarrow\infty} \nabla J^k(x^{k+1})\!+\!A^T\lambda^{k+1}\!=\!0$.
    \label{theo: decrease u}
    \end{theorem}
    \begin{sproof}
     To prove this theorem, we need to show that
    $\Vert u^k-u^k_*\Vert^2_{M^k}\geq \Vert u^{k+1}-u^{k+1}_*\Vert^2_{M^{k+1}}+ \mu^k\Vert u^{k+1}-u^k\Vert^2_{M^k}$ where $\mu^k>0$, which can be obtained by Lemma \ref{lemma: preliminary}. The complete proof is in Section \ref{sec:proof decrease u} in the appendix.
    \end{sproof}
    \indent The condition $\Vert u^{k+1}-u^{k+1}_{*}\Vert^2_{M^{k+1}}\leq \Vert u^{k+1}-u^{k}_{*}\Vert^2_{M^{k}}$ guarantees there exist learnable parameters to make $u^{k+1}$ contractive with respect to the norm induced by $M^k$. Theorem \ref{theo: decrease u} ensures that the gap between $u^{k+1}$ and $u^k$ converges to $0$, and $u^{k+1}$ converges to a feasible solution to Equation \eqref{eq:layer problem}  with constraint satisfaction $\lim\nolimits_{k\rightarrow\infty} Ax^{k+1}-b=0$.
\section{Experiment Verification}
\indent In this section,  nine real-world datasets were utilized to test our proposed IVGD compared with state-of-the-art methods. Performance evaluation, ablation studies, sensitivity analysis, and scalability analysis have demonstrated the effectiveness, robustness, and efficiency of the proposed IVGD. All experiments were conducted on a 64-bit machine with  Intel(R) Xeon(R) quad-core processor (W-2123  CPU@ 3.60 GHZ) and 32.0 GB memory.
\subsection{Experimental Protocols}

\subsubsection{Data Description}
\indent We compare our proposed IVGD with the state-of-the-art methods on nine real-world datasets in the experiments, whose statistics are shown in Table \ref{tab:dataset}.  Due to space limit, their descriptions are outlined in Section \ref{sec:dataset} in the Appendix. The Deezer dataset was only used to evaluate the scalability.\\
\indent For all datasets except the Memetracker and the Digg, we generated diffusion cascades based on the following strategy: $10\%$ nodes were chosen as source nodes randomly, and then the diffusion was repeated 60 times for each
source vector. For each cascade, we have a source vector $x$ and a diffusion vector $Y_T$. For the Memetracker and Digg, they provided true source vectors and diffusion vectors.  The ratio of the sizes of the training set and the test set was 8:2.
\subsubsection{Comparison Methods}
\indent For comparison methods, three state-of-the-art approaches LPSI \cite{wang2017multiple}, NetSleuth \cite{prakash2012spotting} and GCNSI \cite{dong2019multiple} are compared with our proposed IVGD, all of which are outlined as follows:\\
\indent 1. LPSI \cite{wang2017multiple}. LPSI is short for Label Propagation based Source Identification. It is inspired by the label propagation algorithm in semi-supervised learning.\\
\indent 2. NetSleuth \cite{prakash2012spotting}. The goal of NetSleuth is to employ the Minimum
Description Length principle to identify the best set of source nodes and virus propagation ripple \cite{prakash2012spotting}.\\
\indent 3. GCNSI \cite{dong2019multiple}. GCNSI is a Graph Convolutional Network (GCN) based source identification algorithm. It used the LPSI to augment input, and then applied the GCN for source identification.\\
\begin{table}[]
\scriptsize
    \centering
    \begin{tabular}{c|c|c|c|c}
    \hline\hline
         Name&\#Nodes&\#Edges&\tabincell{c}{Average Degree}& Diameter  \\\hline
         Karate& 34&78&2.294&5\\\hline
         Dolphins&62&159&5.129&8\\\hline
          Jazz&198&2,742&13.848&9\\\hline
         Network Science &1,589&2,742&3.451&17\\\hline
         Cora-ML&2,810&7,981&5.68&17\\\hline
         Power Grid&4,941&6,594&2.669&46\\\hline
         Memetracker&1,653&4,267&5.432&4\\\hline
         Digg&11,240&47,885&8.52&4\\\hline
         Deezer&47,538&222,887&9.377&-\\\hline\hline
    \end{tabular}
    \caption{Statistics of nine real-world datasets.}
    \label{tab:dataset}
    \vspace{-0.8cm}
\end{table}
\subsubsection{Metrics}
\indent Five metrics were applied to evaluate the performance: the Accuracy (ACC) is the ratio of accurately labeled nodes to all nodes; the Precision (PR) is the ratio of accurately labeled as source nodes to all nodes labeled as source; the Recall (RE) defines the ratio of accurately labeled as source nodes to all true source nodes; the F-Score (FS) is the harmonic mean of Precision and Recall, which is the most important metric for performance evaluation. This is because the proportion of source nodes is far smaller than that of other nodes (i.e. label imbalance). Besides, the Area Under the Receiver operating characteristic curve (AUC) is an important metric to evaluate a classifier given different thresholds.

\subsubsection{Parameter Settings}
\indent For the proposed IVGD, the $f_W(x)$ was chosen to be a 2-layer MLP model, where the number of hidden units was set to 6 for all datasets except Memetracker and Digg \cite{xia2021DeepIS}. It was set to $100$ and $50$ for Memetracker and Digg, respectively. $g(\zeta)$ was chosen to be the IC function, the number of validity-aware layers was set to 10. The error compensation module was a three-layer MLP model, where the number of neurons was 1,000. $\alpha$, $\tau$ and $\rho$ were set to $1$, $10$ and $10^{-3}$, respectively based on the optimal training performance. The learning rate of the SGD was set to $10^{-3}$. The number of the epoch was set to $100$. The equality constraint we used was $\Vert x\Vert_0=\vert S\vert$, which means that the number of source nodes was known in advance. However, it is non-differentiable and nonconvex. To address this, we relaxed the constraint to a linear constraint as follows:  $\sum\nolimits_{i=1}^n x_i=\vert S\vert$.
This constraint was only used in the training phase.\\
\indent For all comparison methods, the $\alpha$ in LPSI and GCNSI was set to $0.01$ and $0.49$, respectively, based on the optimal training performance. The GCN in GCNSI was a two-layer architecture, where the number of hidden neurons was 128. The learning rate in SGD was set to $10^{-3}$.
\begin{table}[]
\scriptsize
\centering
    \begin{tabular}{c|c|c|c|c|c|c|c|c}
    \hline\hline
    &\multicolumn{4}{c|}{Karate}&\multicolumn{4}{c}{Dolphins} \\
    \hline
         Method&ACC&PR&RE&FS&ACC&PR&RE&FS  \\\hline
         LPSI&0.9559&0.6800&1.0000&0.7970&0.9677&0.7790&1.0000&0.8717\\\hline
         NetSleuth&0.9147&0.5371&0.6833&0.5965&0.9306&0.6454&0.7425&0.6904\\\hline
         GCNSI&0.7088&0.1150&0.2667&0.1581&0.6177&0.1015&0.2548&0.1372\\\hline
         IVGD(ours)&\textbf{0.9853}&\textbf{0.8717}&\textbf{1.0000}&\textbf{0.9213}&\textbf{0.9935}&\textbf{0.9444}&\textbf{1.0000}&\textbf{0.9701}\\\hline
            &\multicolumn{4}{|c|}{Network Science}&\multicolumn{4}{c}{Cora-ML} \\
    \hline
         Method&ACC&PR&RE&FS&ACC&PR&RE&FS  \\
         \hline
         LPSI&0.9831&0.8525&1.0000&0.9202&0.9011&0.5067&0.9993&0.6724\\\hline
         NetSleuth&0.9595&0.7642&0.8429&0.8016&0.8229&0.1627&0.1793&0.1706\\\hline
         GCNSI&0.8840&0.0582&0.0135&0.0218&0.8580&0.0970&0.0478&0.0637\\\hline
         IVGD(ours)&\textbf{0.9946}&\textbf{0.9476}&\textbf{1.0000}&\textbf{0.9730}&\textbf{0.9973}&\textbf{0.9744}&\textbf{1.0000}&\textbf{0.9870}\\
         \hline
         &\multicolumn{4}{|c|}{Jazz}&\multicolumn{4}{c}{Power Grid}\\\hline
         Method&ACC&PR&RE&FS&ACC&PR&RE&FS\\\hline
         LPSI&0.9035&0.6074&0.9944&0.7371&0.9673&0.7584&1.0000&0.8624\\\hline
         NetSleuth&0.9222&0.5904&0.6629&0.6245&0.9276&0.6347&0.6986&0.6651\\\hline
         GCNSI&0.7525&0.0685&0.1280&0.0849&0.7125&0.1022&0.2285&0.1410\\\hline
         IVGD(ours)&\textbf{0.9980}&\textbf{0.9802}&\textbf{1.0000}&\textbf{0.9899}&\textbf{0.9902}&\textbf{0.9133}&\textbf{1.0000}&\textbf{0.9546}\\
         
         \hline\hline
    \end{tabular}
    \caption{Test performance of simulations on six datasets: the proposed IVGD dominates in all methods on six datasets.}
    \label{tab: synthetic performance}
    \vspace{-0.8cm}
\end{table}
\begin{figure}
\begin{minipage}{0.49\linewidth}
  \centering
    \includegraphics[width=\linewidth]{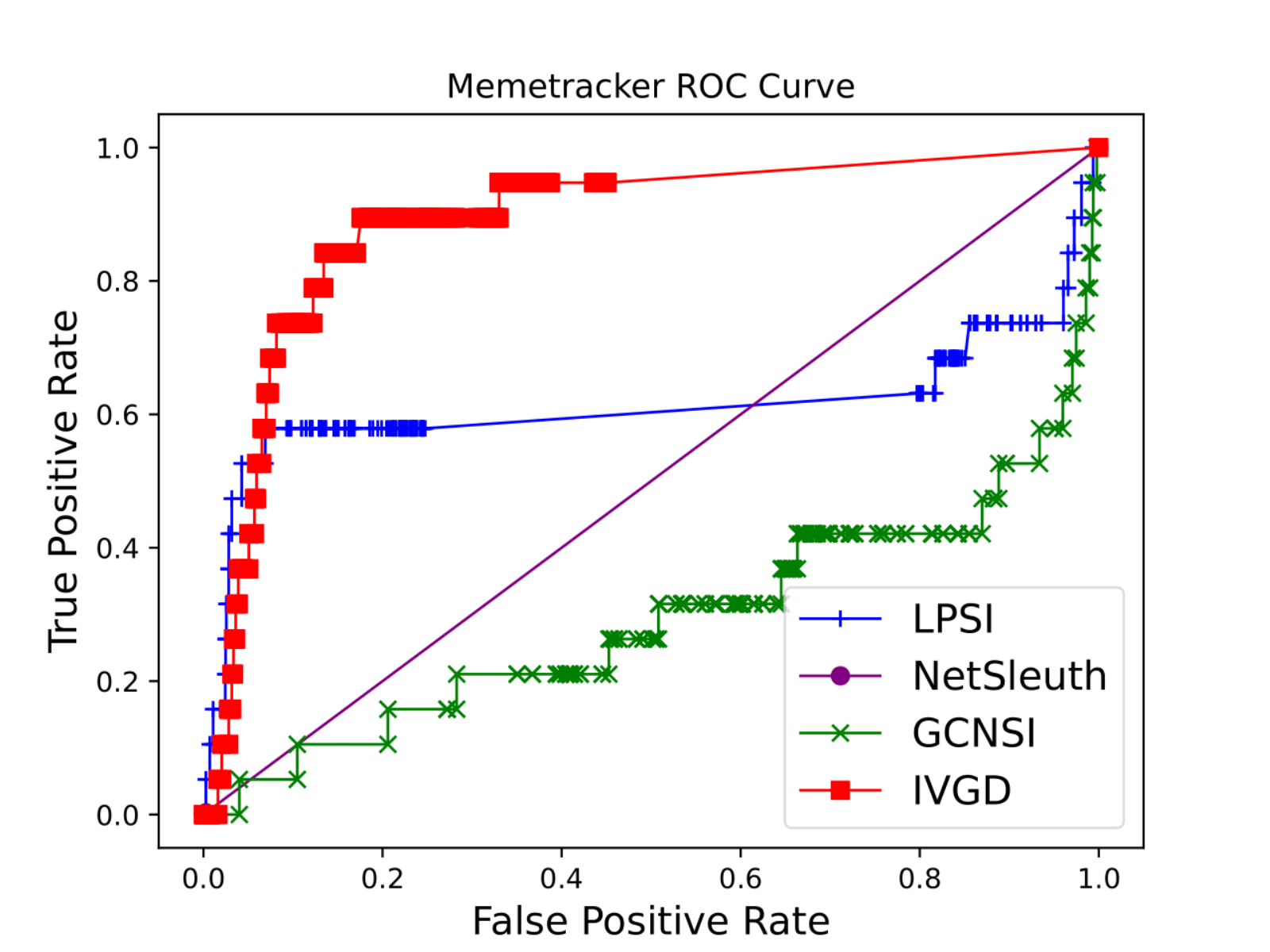}
    \centerline{(a). Memetracker}
\end{minipage}
\hfill
\begin{minipage}{0.49\linewidth}
  \centering
    \includegraphics[width=\linewidth]{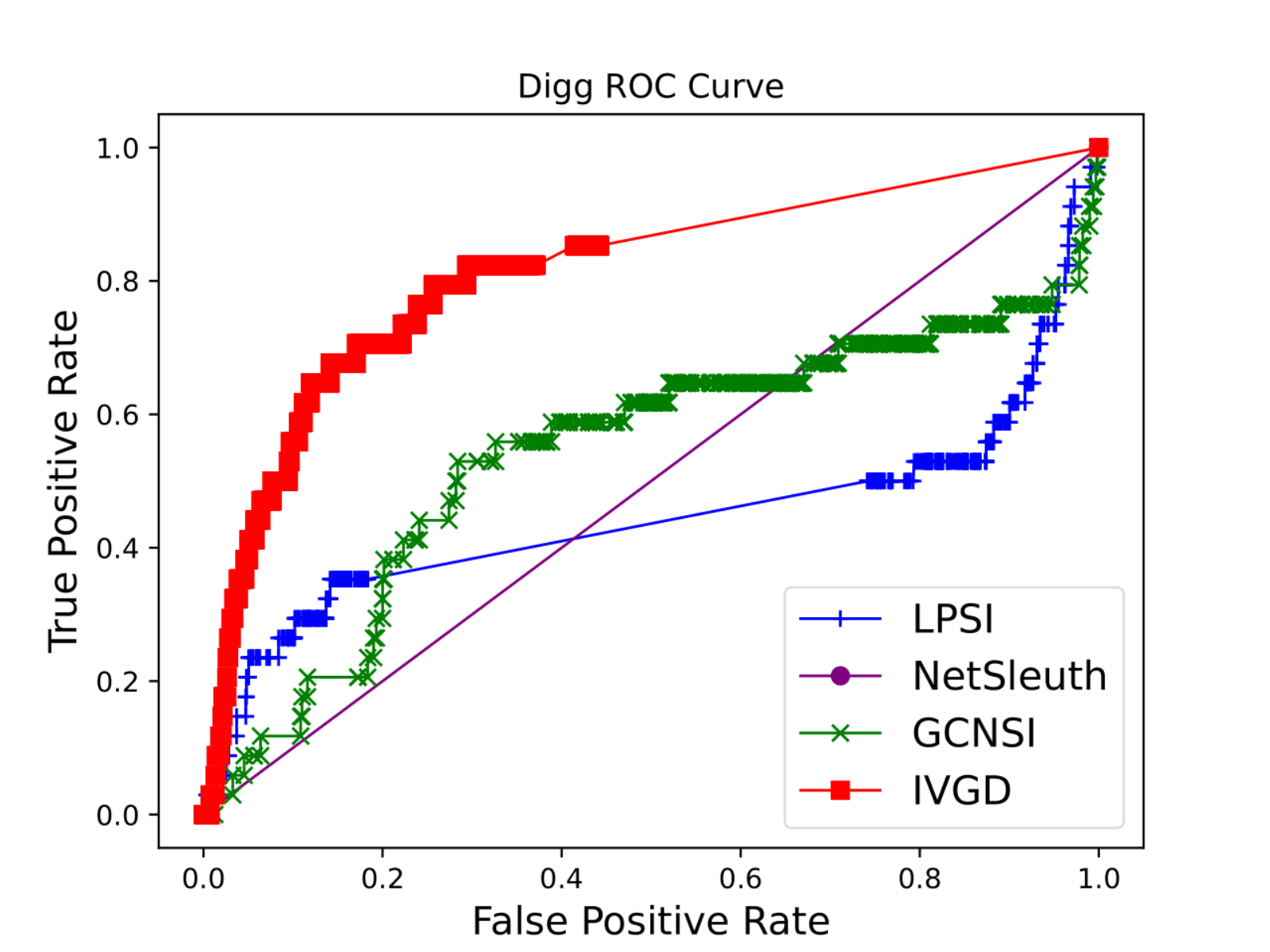}
        \centerline{(b). Digg}
\end{minipage}
        \caption{ROC curves on two real-world datasets: all comparison methods are surrounded by the proposed IVGD.}
    \label{fig:ROC}
    \vspace{-0.4cm}
\end{figure}

\subsection{Performance Evaluation}
\indent The test performance of all methods on six datasets is demonstrated in Table \ref{tab: synthetic performance}.  The best performance is highlighted in bold. Overall, our proposed IVGD outperforms all comparison methods significantly on all datasets. Specifically, the ACCs of the proposed IVGD on six datasets are all above $0.98$, and the PRs are also above $0.9$.  Most importantly, the FS metric of the proposed IVGD is in the vicinity of $0.96$ on average. This substantiates that our proposed IVGD algorithm can accurately predict source nodes despite their scarcity. For comparison methods, the LPSI performs the best followed by NetSleuth: the ACCs of the LPSI are $2\%$ higher than those of NetSleuth, and the PRs are around $10\%$ better. For example, our proposed IVGD attains 0.99 and 0.95 in the ACC and the PR on the Power Grid dataset, respectively, while the counterparts of the LPSI are 0.97 and 0.86, respectively, and the NetSleuth achieves 0.62 and 0.65, respectively. The GCNSI performs poorly on all datasets: its PRs and FSes are surprisingly low, which are below 0.3. For instance, the PR and the FS on the Network Science dataset are 0.08 and 0.03, respectively. This is because it cannot differentiate source nodes from others, and its predictions are in the vicinity of the threshold. This demonstrates that the GCNSI cannot draw a clear decision boundary.
\begin{table}[]
\scriptsize
\centering
    \begin{tabular}{c|c|c|c|c|c|c|c|c}
    \hline\hline
    &\multicolumn{4}{c|}{Karate}&\multicolumn{4}{c}{Dolphins} \\
    \hline
         Method&ACC&PR&RE&FS&ACC&PR&RE&FS  \\\hline
         IVGD(1)&0.9618&0.7167&1.0000&0.8248&0.9694&0.7873&1.0000&0.8774\\\hline
         IVGD(2)&\textbf{0.9882}&\textbf{0.8967}&\textbf{1.0000}&\textbf{0.9356}&0.9726&0.8070&1.0000&0.8904\\\hline
         IVGD(3)&0.9500&0.6583&1.0000&0.7824&0.9613&0.7519&1.0000&0.8517\\\hline
         IVGD&0.9853&0.8717&1.0000&0.9213&\textbf{0.9935}&\textbf{0.9444}&\textbf{1.0000}&\textbf{0.9701}\\\hline
            &\multicolumn{4}{c|}{Network Science}&\multicolumn{4}{c}{Cora-ML}\\\hline
         Method&ACC&PR&RE&FS&ACC&PR&RE&FS   \\
         \hline
         IVGD(1)&0.9859&0.8736&1.0000&0.9324&0.8712&0.4411&1.0000&0.6121\\
         \hline
         IVGD(2)&0.9867&0.8799&1.0000&0.9360&0.9959&0.9617&1.0000&0.9805\\\hline
         IVGD(3)&0.9812&0.8383&1.0000&0.9119&0.9850&0.8717&1.0000&0.9314\\\hline
         IVGD&\textbf{0.9946}&\textbf{0.9476}&\textbf{1.0000}&\textbf{0.9730}&\textbf{0.9973}&\textbf{0.9744}&\textbf{1.0000}&\textbf{0.9870}\\
         \hline
         &\multicolumn{4}{c}{Jazz}&\multicolumn{4}{|c}{Power Grid}\\
         \hline
         Method&ACC&PR&RE&FS &ACC&PR&RE&FS\\
         \hline
         IVGD(1)&0.9586&0.7092&1.0000&0.8266&0.9729&0.7916&1.0000&0.8835\\
         \hline
         IVGD(2)&0.9904&0.9123&1.0000&0.9528&0.9881&0.8967&1.0000&0.9455\\
         \hline
         IVGD(3)&0.9646&0.7394&1.0000&0.8473&0.9631&0.7358&1.0000&0.8475\\
         \hline
         IVGD&\textbf{0.9980}&\textbf{0.9802}&\textbf{1.0000}&\textbf{0.9899}&\textbf{0.9902}&\textbf{0.9133}&\textbf{1.0000}&\textbf{0.9546}\\
         \hline\hline
    \end{tabular}
    \caption{Ablation studies on simulations of six test datasets: all components in our model contribute to the outstanding performance.}
\label{tab:alabtion}
\vspace{-1.1cm}
\end{table}
\\ \indent Aside from simulations, we also evaluate our proposed IVGD on two real-world datasets, as shown in Figure \ref{fig:ROC}. X-axis and Y-axis represent the true positive rate and the false positive rate, respectively. Similarly as Table \ref{tab: synthetic performance}, our proposed IVGD also outperforms others significantly on the ROC curves: all comparison methods are surrounded by our proposed IVGD.  Specifically, the AUCs of our proposed IVGD are above $0.6$ on the Memetracker and the Digg datasets, whereas these of all other methods are either around $0.5$ or below $0.5$. The LPSI outperforms the GCNSI by about $20\%$ on the Memetracker, while it performs worse on the Digg by $10\%$.
\subsection{Ablation Studies}
\indent One important question to our proposed IVGD is whether all components in our model are necessary. To investigate this, we test our performance on six datasets, when some components are removed. For the sake of simplicity, IVGD(1) means that the invertible graph residual net was removed, IVGD(2) means that the error compensation module was removed, and IVGD(3) means that validity-aware layers were removed. The results on six datasets are demonstrated in Table \ref{tab:alabtion}. Overall, the performance will degrade if any component of our proposed IVGD is removed. Without the invertible graph residual net, the performance on the Cora-ML dataset drops significantly from $99.7\%$ to $87\%$ in the ACC, and the PR has declined by more than $50\%$. The FS on the Power Grid has demonstrated a $7\%$ drop due to the same reason. Validity-aware layers provide a giant leap on FS when we compare IVGD(3) with IVGD.  The performance of the FS has been enhanced by $5\%-14\%$. For example, the FS on the Power Grid dataset soars from $0.848$ to $0.955$. The same pattern is applied to the Karate and the Dolphins datasets. This suggests that integrating validity patterns into the model significantly improves model performance. The error compensation module plays a less important role than the invertible graph residual net and validity-aware layers. Specifically, the performance drop without it is slim compared with other components. For example, the drop of the ACC on the Jazz dataset is less than $1\%$. Moreover, the ACC on the Karate dataset even increases slightly. But the effect of the error compensation module is still positive overall.
\begin{figure}
\begin{minipage}{0.32\linewidth}
   \centering
    \includegraphics[width=\linewidth]{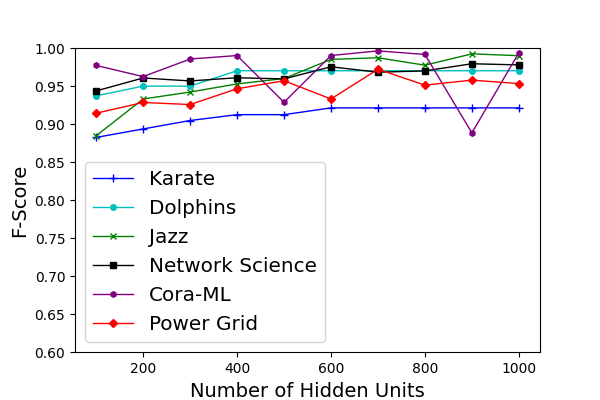}
    \centerline{(a). Hidden units VS FS.}
\end{minipage}
    \hfill
\begin{minipage}{0.32\linewidth}
   \centering
    \includegraphics[width=\linewidth]{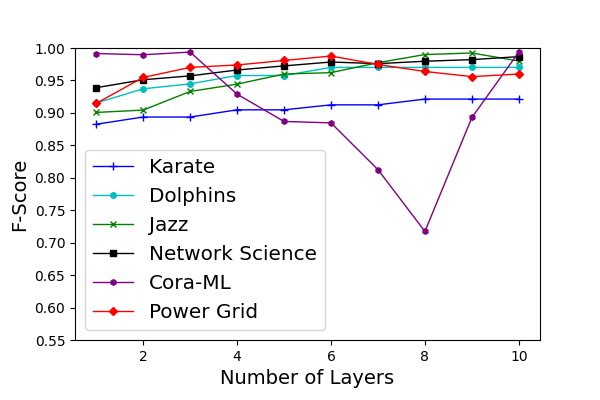}
    \centerline{(b). Layers VS FS.}
\end{minipage}    
\hfill
\begin{minipage}{0.32\linewidth}
  \centering
    \includegraphics[width=\linewidth]{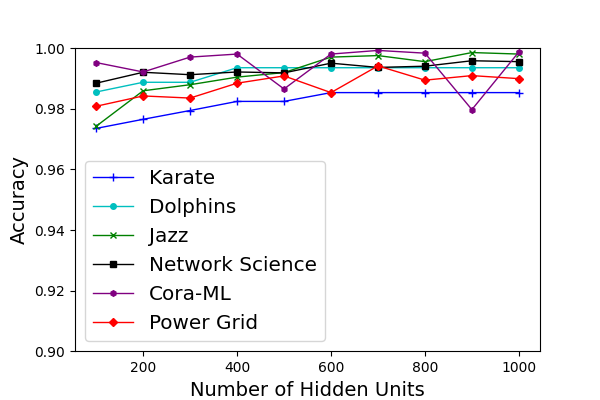}
    \centerline{(c). Hidden units VS ACC.}
\end{minipage}
\vfill
\begin{minipage}{0.32\linewidth}
  \centering
    \includegraphics[width=\linewidth]{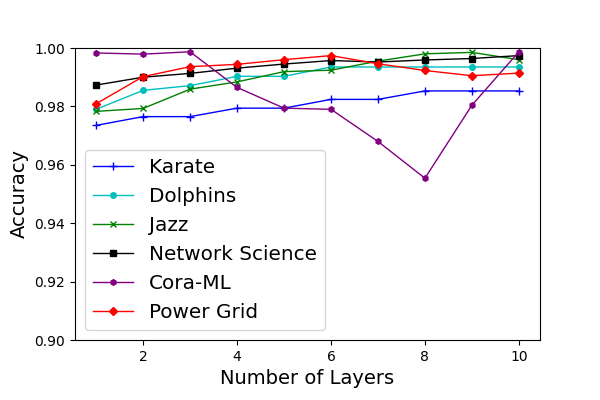}  
         \centerline{(d). Layers VS ACC.}
\end{minipage}
\hfill
\begin{minipage}{0.32\linewidth}
  \centering
    \includegraphics[width=\linewidth]{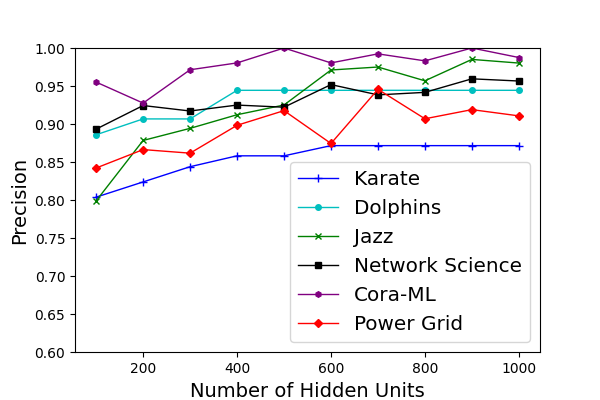}
    \centerline{(e). Hidden units VS PR.}
\end{minipage}
\hfill
\begin{minipage}{0.32\linewidth}
  \centering
    \includegraphics[width=\linewidth]{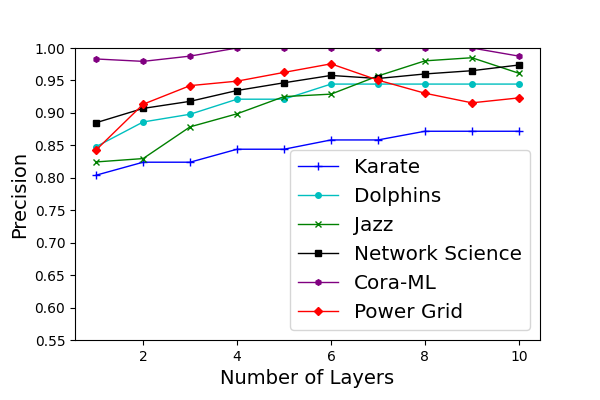}  
        \centerline{(f). Layers VS PR.}
\end{minipage}
\vspace{-0.3cm}
\caption{The impacts of two factors on the FS, ACC and PR: more hidden units and layers lead to better performance.}
\vspace{-0.3cm}
\label{fig:impact}
\end{figure}
    \begin{figure*}[!t]
		\subfloat[LPSI]{\label{fig: lpsi_karate}
			\includegraphics[width=0.19\textwidth]{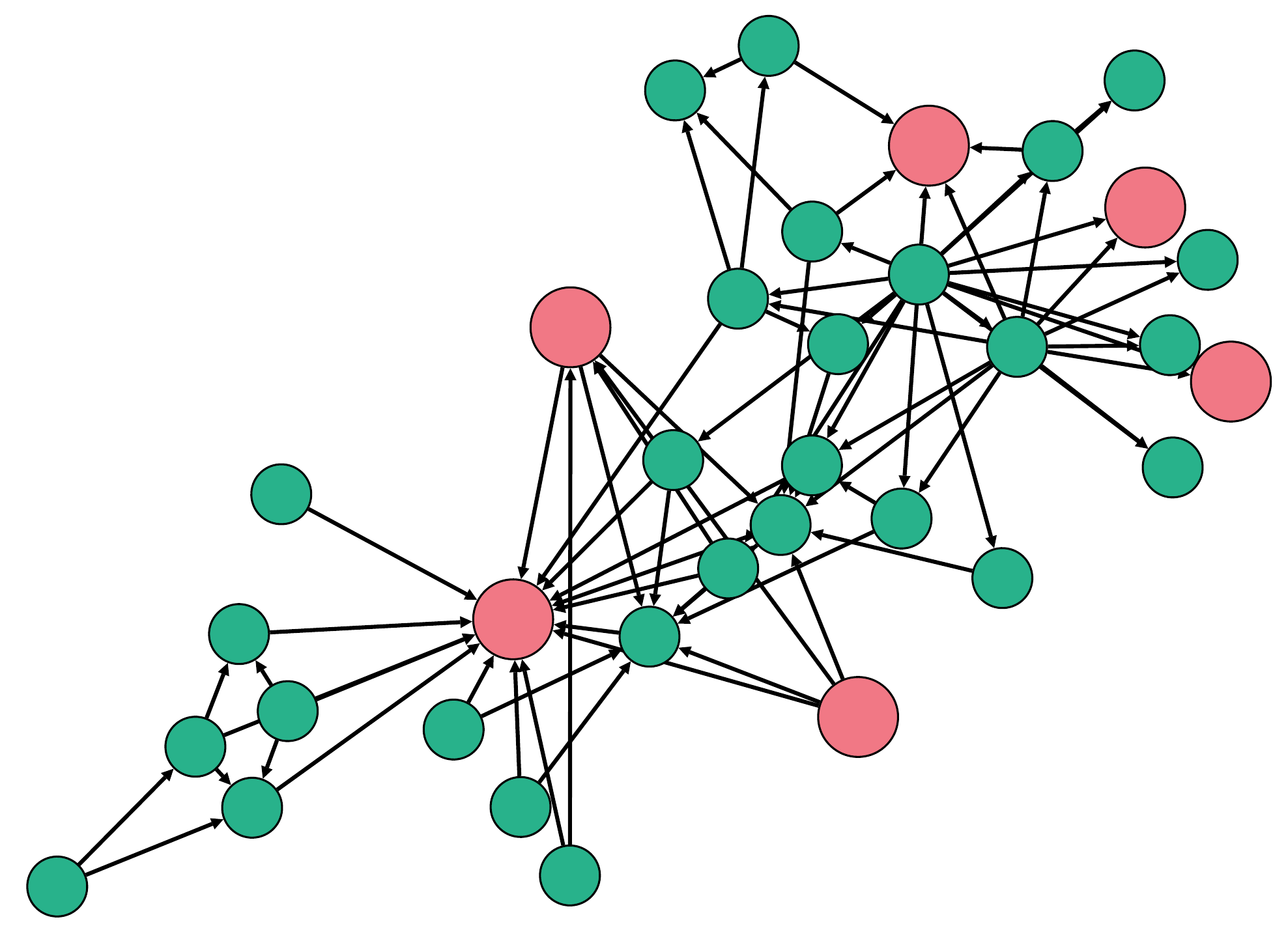}}
		\subfloat[NetSleuth]{\label{fig: NetSleuth_karate}
			\includegraphics[width=0.19\textwidth]{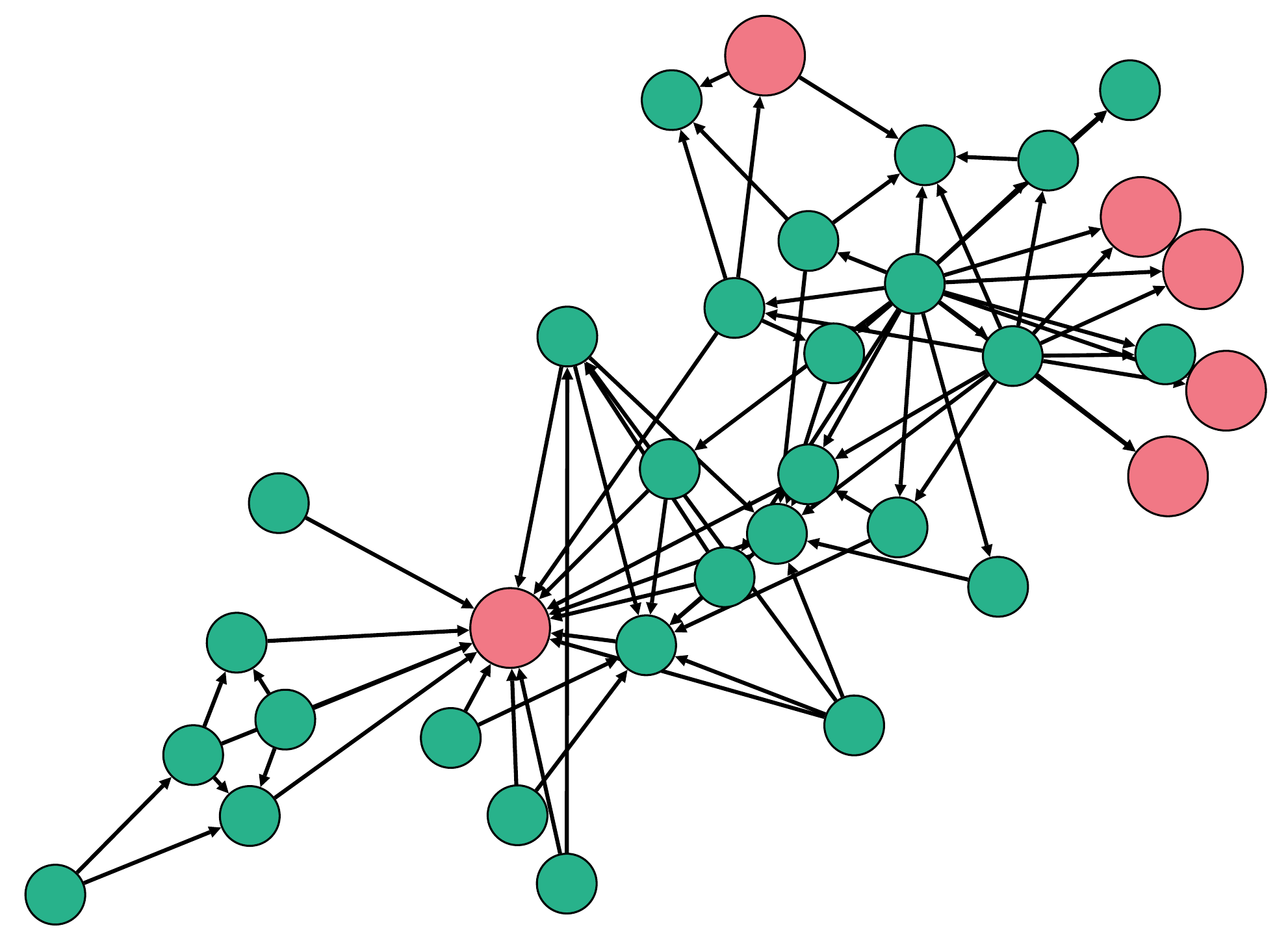}}
		\subfloat[GCNSI]{\label{fig: gcnsi_karate}
			\includegraphics[width=0.19\textwidth]{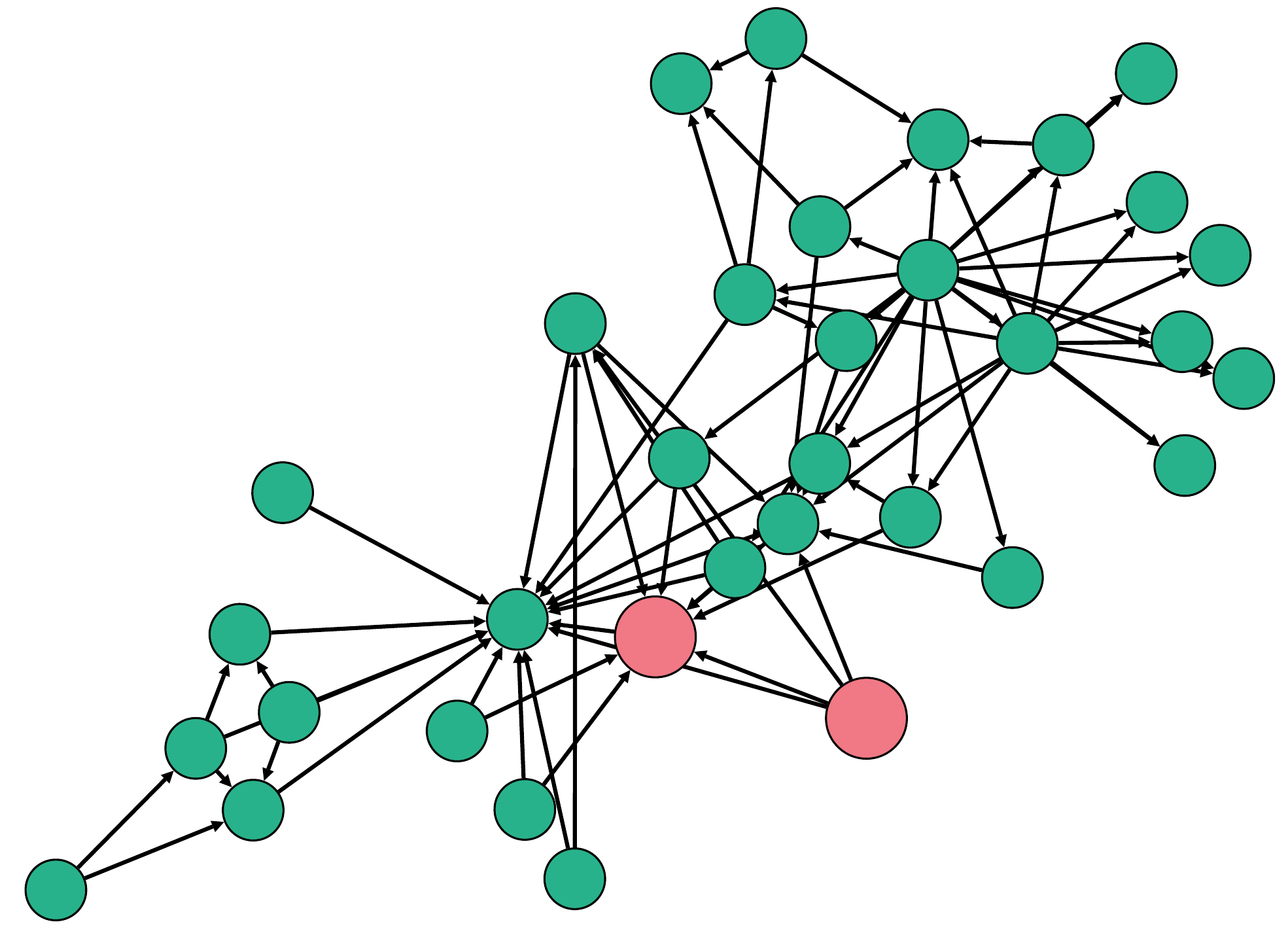}}
		\subfloat[\textbf{IVGD}]{\label{fig: ivgd_karate}
			\includegraphics[width=0.19\textwidth]{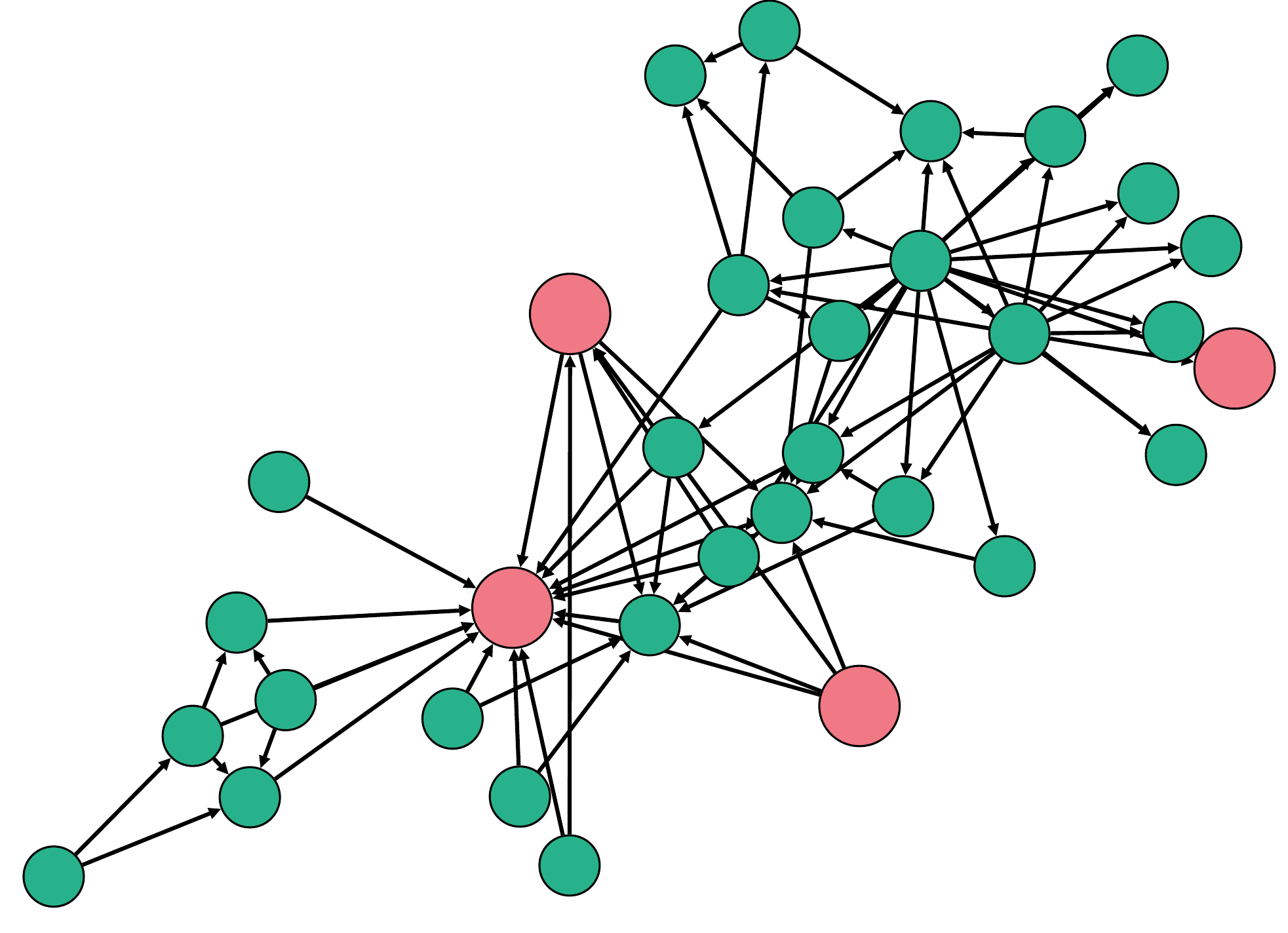}}
		\subfloat[True Sources]{\label{fig: true_karate}
			\includegraphics[width=0.19\textwidth]{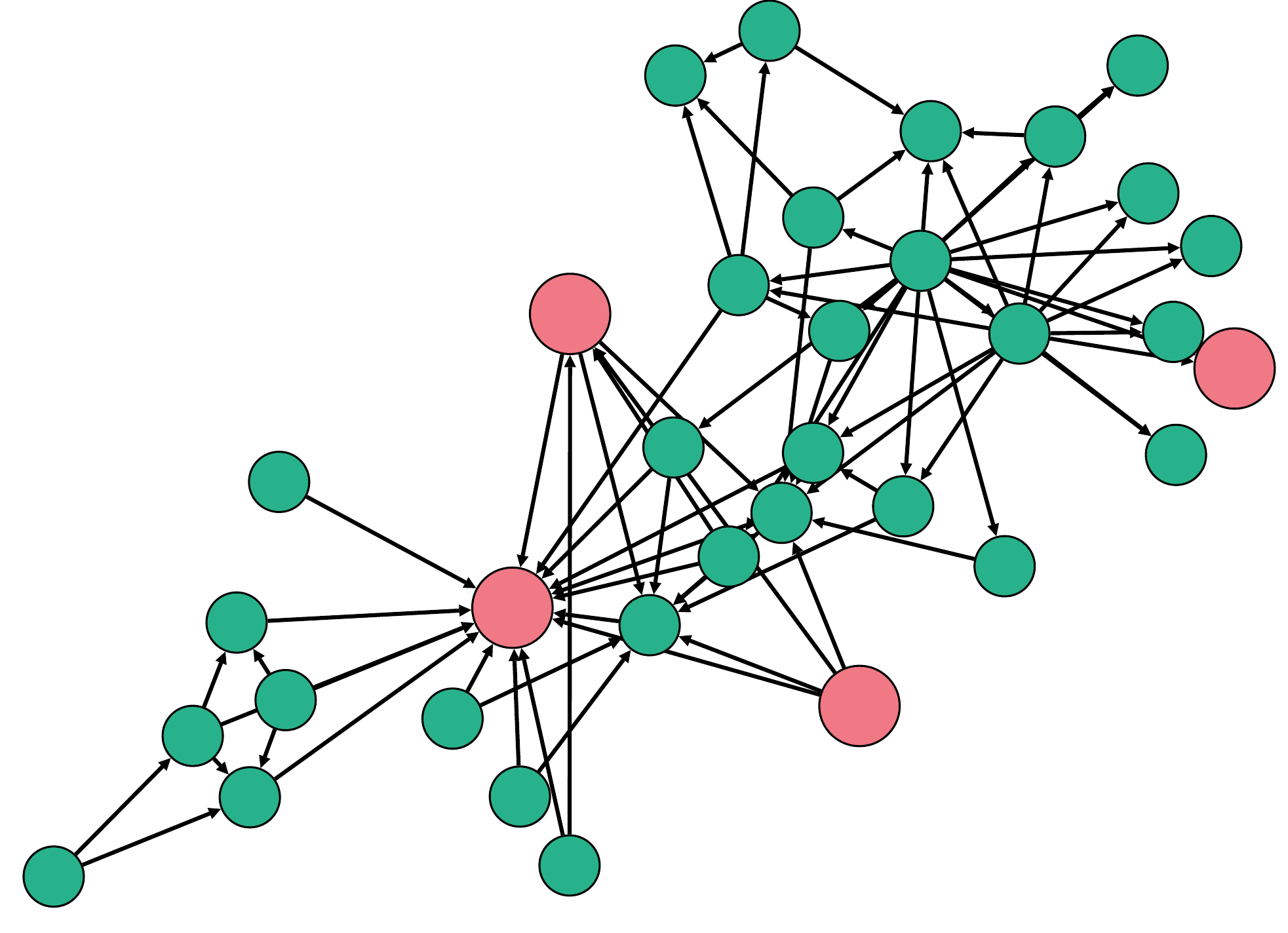}}\\\vspace{-4mm}
	\subfloat[LPSI]{\label{fig: lpsi_dolphins}
			\includegraphics[width=0.19\textwidth]{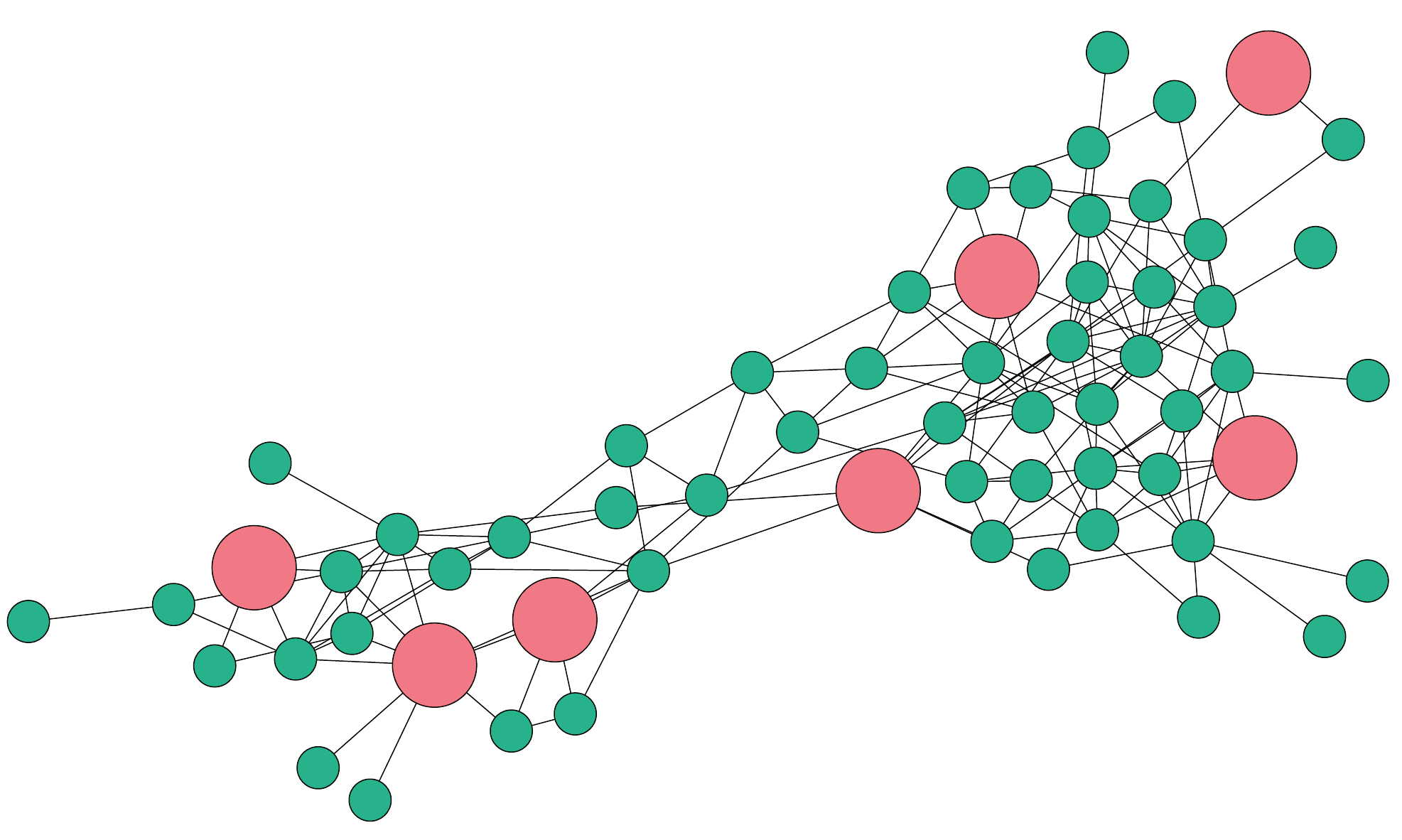}}
		\subfloat[NetSleuth]{\label{fig: NetSleuth_dolphins}
			\includegraphics[width=0.19\textwidth]{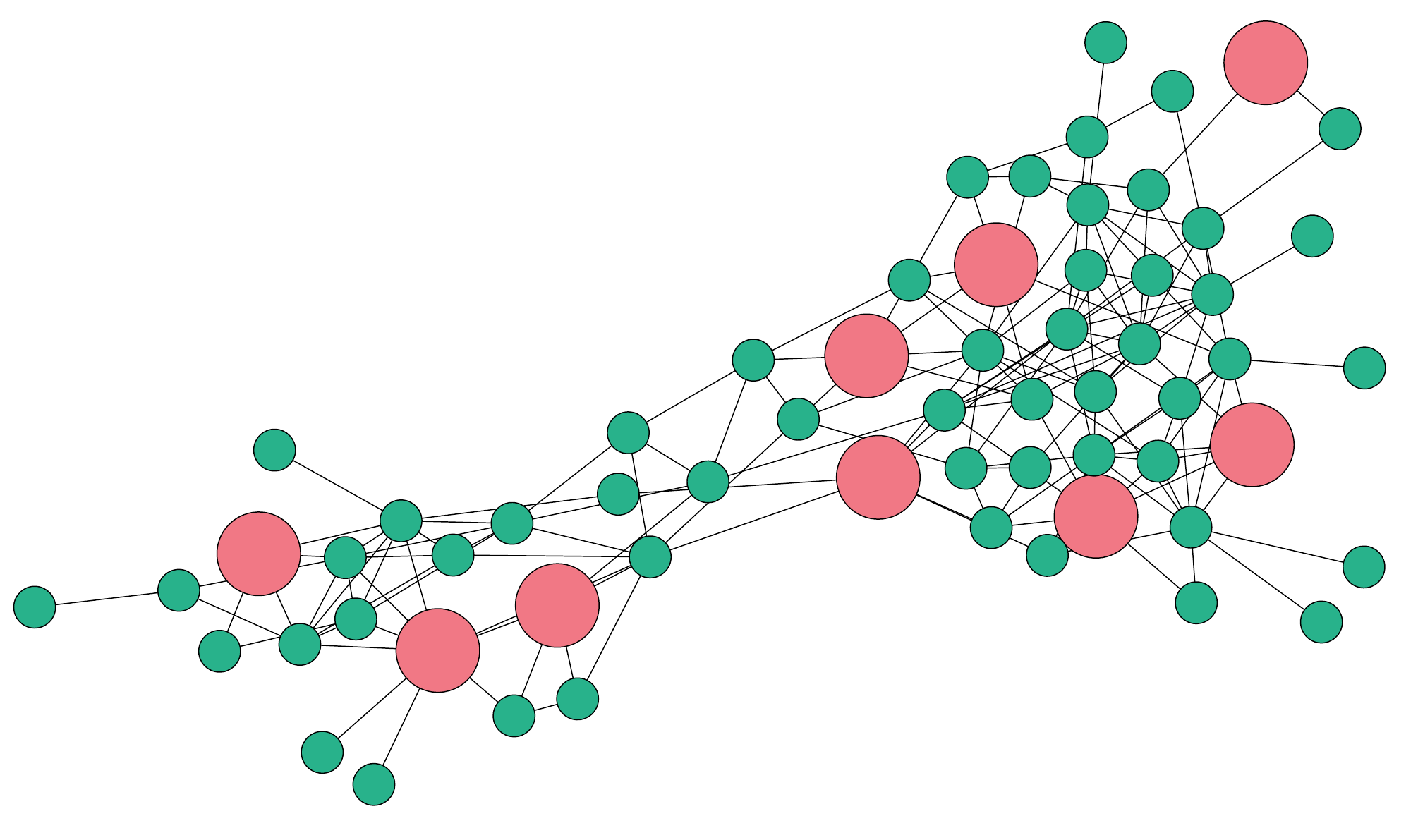}}
		\subfloat[GCNSI]{\label{fig: gcnsi_dolphins}
			\includegraphics[width=0.19\textwidth]{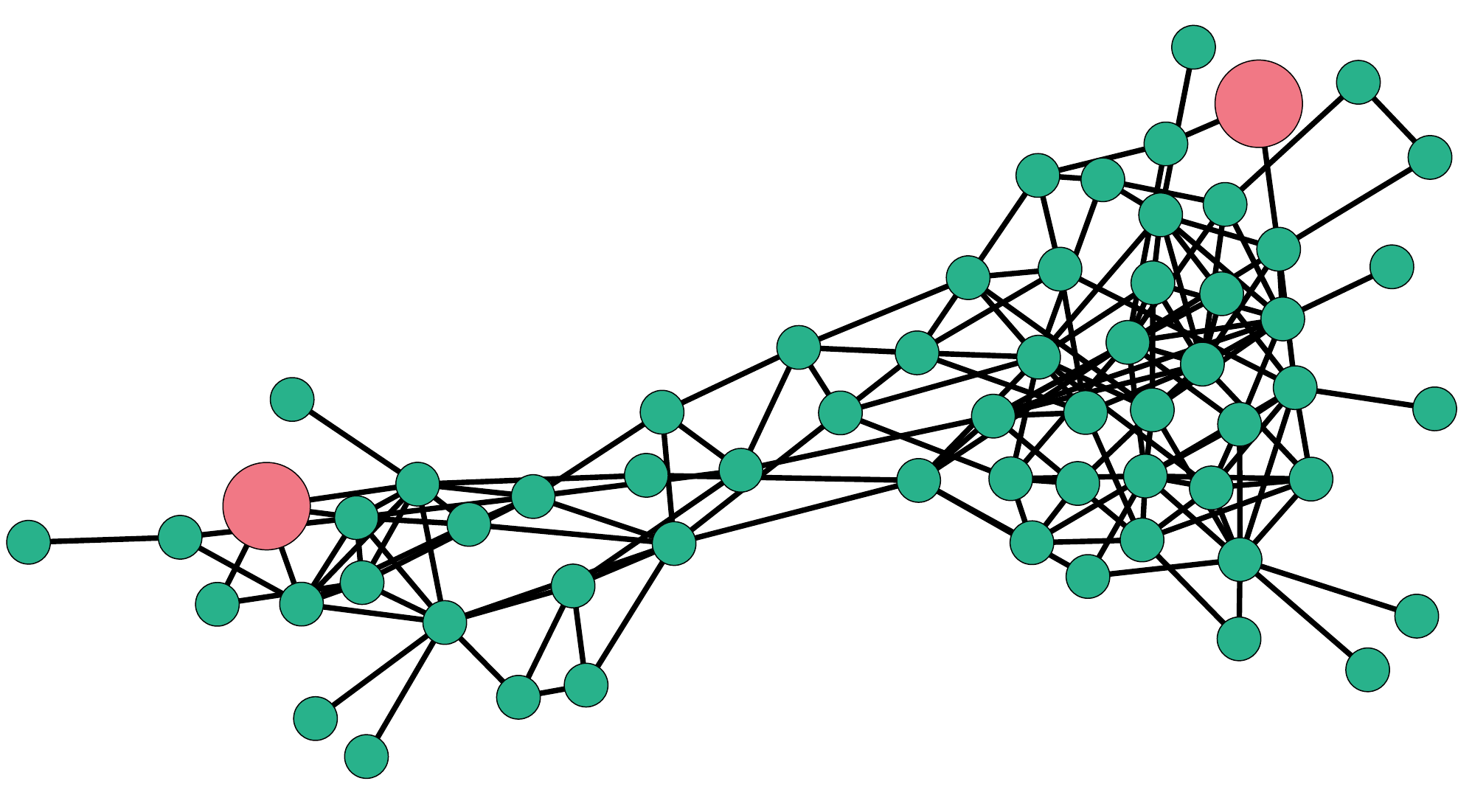}}
		\subfloat[\textbf{IVGD}]{\label{fig: ivgd_dolphins}
			\includegraphics[width=0.19\textwidth]{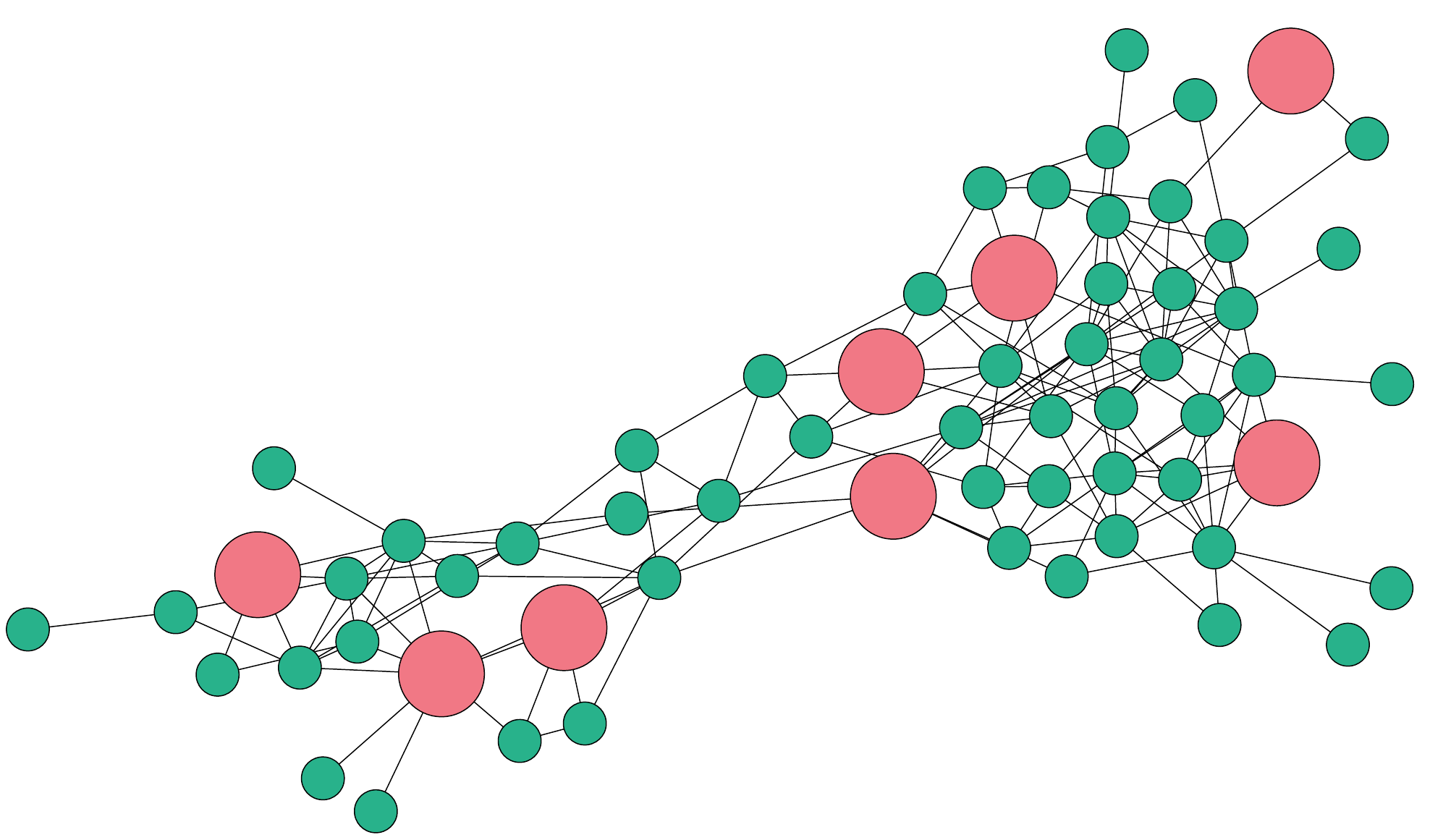}}
		\subfloat[True Sources]{\label{fig: true_dolphins}
			\includegraphics[width=0.19\textwidth]{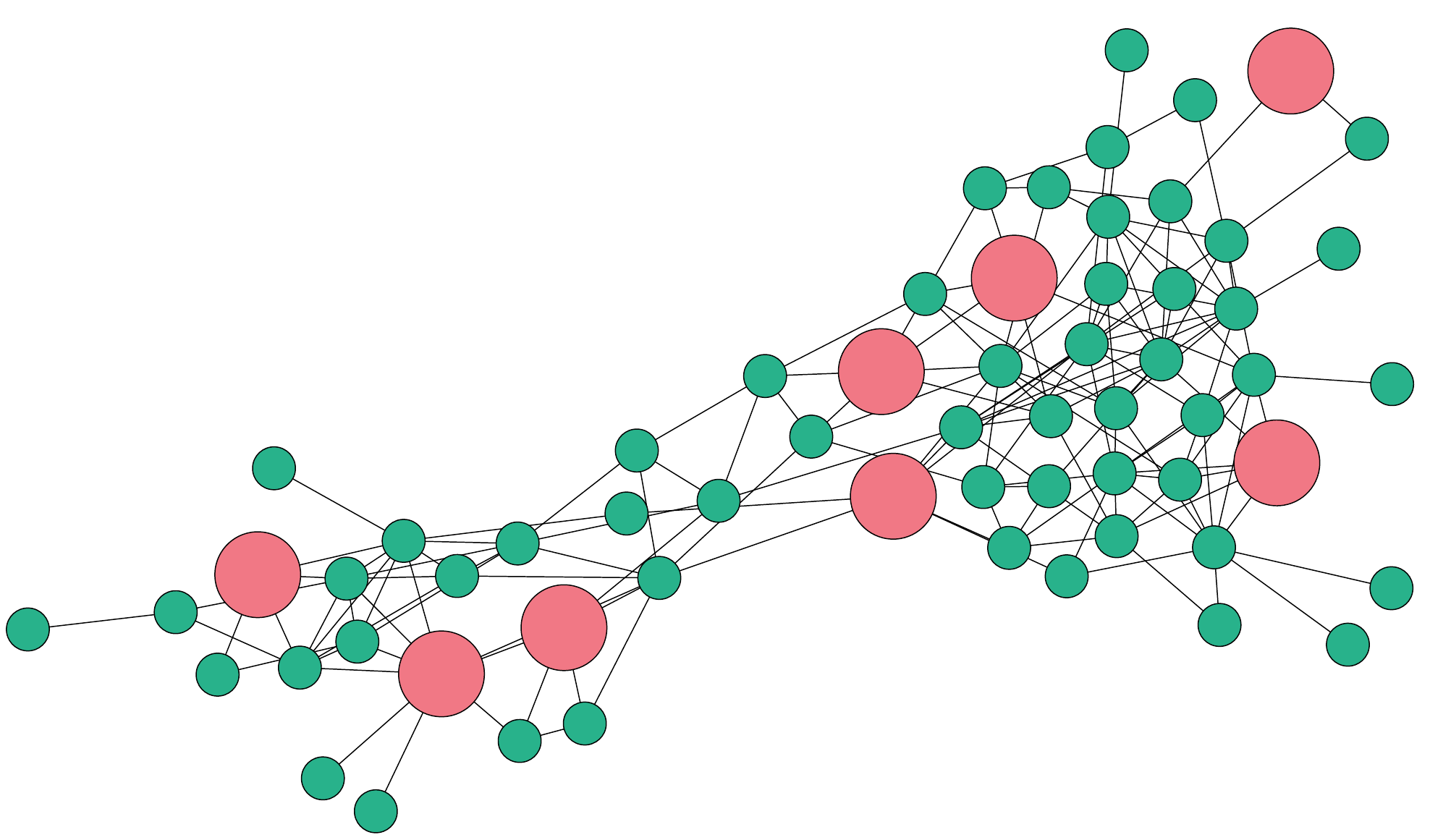}}\\\vspace{-3mm}
		\caption{ Visualizations of two datasets for all methods and true source patterns. Figures \ref{fig: lpsi_karate} - \ref{fig: true_karate} are visualizations of Karate, and Figures \ref{fig: lpsi_dolphins} - \ref{fig: true_dolphins} are visualizations of Dolphins. Sources nodes and other nodes are marked with red and green, respectively.}
		\label{fig:visualization}
	\end{figure*}

\subsection{Sensitivity Analysis}
\indent Next, it is crucial to investigate how parameter settings affect performance.  In this section, we explore two factors: the number of hidden units in the compensation module and the number of validity-aware layers. The number of epochs was set to 10. For the hidden units, we changed the number from $100$ to $1,000$; for the layers, the number ranged from 1 to 10. The impacts on FS, ACC, and PR are shown in Figure \ref{fig:impact}. Overall, the performance increases smoothly with the increase of hidden units and layers. For example, the FS on the Network Science dataset increases by $2\%$ when hidden units are changed from $100$ to $1,000$; it climbs by $4\%$ when the layers ranged from $1$ to $10$.  However, there is an exception: the FS fluctuates on the Cora-ML dataset. The amplitudes are about $10\%$ and $30\%$ for hidden units and layers, respectively. Despite the fluctuation, the lowest FSes achieved by 800 hidden units and seven layers are still better than all comparison methods, as shown in Table \ref{tab: synthetic performance}.
\subsection{Scalability Analysis}
\begin{table}[]
\scriptsize
\centering
    \begin{tabular}{c|c|c|c|c|c|c|c}
        \hline\hline
         Method&Karate&Dolphins&Jazz&\tabincell{c}{Network \\Science}&Cora-ML&\tabincell{c}{Power\\ Grid}&Deezer \\\hline
         LPSI&\textbf{0.26}&\textbf{0.27}&\textbf{0.76}&52.83&240.88&899.45&94541.13\\\hline
         NetSleuth&0.33&0.48&1.95&32.96&645.04&1260.68&114425.32\\\hline
         GCNSI&1.55&34.62&125.59&283.62&776.70&2324.53&174923.31\\\hline
         IVGD(ours) &5.37&6.45&9.78&\textbf{23.95}&\textbf{92.68}&\textbf{177.32}&\textbf{46832}\\\hline
    \hline
        \end{tabular}
    \caption{The running time (seconds) on simulations of seven datasets: our proposed IVGD runs the most efficiently on the large-scale networks.}
    \label{tab:running time}
    \vspace{-1cm}
\end{table}
\indent To test the efficiency and scalability of our proposed IVGD, we compared the running time of IVGD with all comparison methods on seven datasets, which is shown in Table \ref{tab:running time}. The best running time is highlighted in bold. In general, we proposed IVGD runs the most efficiently on large-scale networks such as Deezer, which consists of about $5,0000$  nodes. Specifically, it consumes about half a day to finish training, while all comparison methods at least double. The same trend holds in other large networks such as Cora-ML and Power Grid. The LPSI takes the least time on small networks such as Karate and Dolphins. The GCNSI is the slowest method on most datasets. For example, it consumes around 2 minutes on the small Jazz dataset, while all other methods take less than 10 seconds. It takes 2 days on the Deezer dataset, whereas the LPSI only requires half of that time.
\subsection{Invertibility Analysis}
\indent For the invertibility, one may raise a concern on whether it impairs the performance of graph diffusion models. To investigate this question, we compare the performance of the GNN model $\theta$ and the proposed Invertible Graph Residual Net (IGRN) $P$ on eight datasets. The DeepIS \cite{xia2021DeepIS} was chosen as the GNN model, and the IGRN was implemented based on the DeepIS. The Mean Square Error (MSE) and the Mean Absolute Error (MAE) were used to assess the performance. Table \ref{tab:invertibility} illustrates the performance of two graph diffusion models. In summary, they perform similarly on two metrics across different datasets. Specifically,  the GNN achieves a better performance on the Karate, Dolphins, Cora-ML, and Power Grid datasets, whereas the IGRN stands out on the Dolphins, Network Science, Memetracker, and Digg datasets. The largest gap comes from the MAE on the Karate dataset, where the GNN outperforms IGRN by $0.02$.
\begin{table}[]
\scriptsize
    \centering
    \begin{tabular}{c|c|c|c|c|c|c|c|c}
    \hline\hline
         &\multicolumn{2}{c|}{Karate}&\multicolumn{2}{c|}{Dolphins}&\multicolumn{2}{c|}{Jazz}&\multicolumn{2}{c}{Network Science}\\
         \hline
         &MSE&MAE&MSE&MAE&MSE&MAE&MSE&MAE\\\hline
         GNN&\textbf{0.0287}&\textbf{0.0773}&0.0270&0.1063&0.0575&\textbf{0.1731}&0.0199&0.0743\\
        \hline IGRN&0.0311&0.1010&\textbf{0.0258}&\textbf{0.0794}&\textbf{0.0514}&0.1867&\textbf{0.0156}&\textbf{0.0643}\\
        \hline &\multicolumn{2}{c|}{Cora-ML}&\multicolumn{2}{c|}{Power Grid}&\multicolumn{2}{c|}{Memetracker}&\multicolumn{2}{c}{Digg}\\
         \hline
         &MSE&MAE&MSE&MAE&MSE&MAE&MSE&MAE\\\hline
         GNN&\textbf{0.0017}&\textbf{0.0282}&\textbf{0.0221}&\textbf{0.0633}&0.0273&0.0322&0.0198&0.0265\\
        \hline IGRN&0.0041&0.0488&0.0247&0.0758&\textbf{0.0236}&\textbf{0.0311}&\textbf{0.0165}&\textbf{0.0203}\\         \hline\hline
    \end{tabular}
    \caption{The effect of invertibility on graph diffusion models: it plays a negligible role.}
    \label{tab:invertibility}
    \vspace{-1.0cm}
\end{table}
\subsection{Visualization}
\indent Finally, we demonstrate the effectiveness of our proposed IVGD by visualizing two small datasets Karate and Dolphin in Figure \ref{fig:visualization}. Red nodes and green nodes represent source nodes and other nodes, respectively. Specifically, our proposed IVGD perfectly predicts all sources on two datasets, and the LPSI and the NetSleuth also achieve similar source patterns as the ground truth: they only misclassify several source nodes. The GCNSI, however, misses most of the source nodes. This is because it suffers from class imbalance problems, and tends to classify none of all nodes as a source node. This is consistent with test performance shown in Table \ref{tab: synthetic performance}.
\section{Conclusion}
\indent Graph source localization is an important yet challenging problem in graph mining. In this paper, we propose a novel Invertible Validity-aware Graph Diffusion (IVGD) to address this problem from the perspective of the inverse problem. Firstly, we propose an invertible graph residual net by restricting its Lipschitz constant with guarantees. Moreover, we present an error compensation module to reduce the introduced errors with skip connection. Finally, we utilize the unrolled optimization technique to impose validity constraints on the model. A linearization technique is used to transform problems into solvable forms. We provide the convergence of the proposed IVGD to a feasible solution. Extensive experiments on nine real-world datasets have demonstrated the effectiveness, robustness, and efficiency of our proposed IVGD.
\section*{Acknowledgement}
We would like to acknowledge our collaborator, Hongyi Li from Xidian University, for her precious suggestions on the mathematical proofs. Furthermore, this work was supported by the National Science Foundation (NSF) Grant No. 1755850, No. 1841520, No. 2007716, No. 2007976, No. 1942594, No. 1907805, a Jeffress Memorial Trust Award, Amazon Research Award, NVIDIA GPU Grant, and Design Knowledge Company (subcontract No: 10827.002.120.04).
\bibliographystyle{plain}
\bibliography{reference}
\newpage
\onecolumn
\appendix
\small
\textbf{Appendix}\\
\section{The proof of Lemma \ref{lemma:lipschitz constant}}
\label{sec:proof lipschitz constant}
\begin{proof}
On one hand, for any $x^{'},x^{''}$, we have
\begin{align*}
    &\Vert P(x^{''})-P(x^{'})\Vert =\Vert G(F_W(x^{''}))-G(F_W(x^{'}))\Vert=\Vert \frac{g(F_W(x^{''}))+F_W(x^{''})}{2}-\frac{g(F_W(x^{'}))+F_W(x^{'})}{2}\Vert\\&\leq\frac{1}{2}\Vert g(F_W(x^{''}))-g(F_W(x^{'}))\Vert+\frac{1}{2}\Vert F_W(x^{''})-F_W(x^{'})\Vert \ (\text{triangle inequality})\\&\leq \frac{L_g+1}{2}\Vert F_W(x^{''})-F_W(x^{'})\Vert \ (\text{Lipschitz constant of $g$})\\&=\frac{L_g+1}{4}\Vert f_W(x^{''})+x^{''}-f_W(x^{'})-x^{'}\Vert\\&\leq \frac{L_g+1}{4}(\Vert f_W(x^{''})-f_W(x^{'})\Vert+\Vert x^{''}-x^{'}\Vert) \ \text{(triangle inequality)}\\&\leq \frac{(L_f+1)(L_g+1)}{4}\Vert x^{''}-x^{'}\Vert \ (\text{Lipschitz constant of $f_w$}).
\end{align*}
This suggests that $L_P\leq \frac{(L_f+1)(L_g+1)}{4}$. On the other hand,
\begin{align*}
    &\Vert P(x^{''})-P(x^{'})\Vert=\Vert \frac{g(F_W(x^{''}))+F_W(x^{''})}{2}-\frac{g(F_W(x^{'}))+F_W(x^{'})}{2}\Vert\\&\geq \frac{1}{2}\Vert F_W(x^{''})-F_W(x^{'})\Vert-\frac{1}{2}\Vert g(F_W(x^{''}))-g(F_W(x^{'}))\Vert \ (\text{triangle inequality})\\&\geq \frac{1-L_g}{2}\Vert F_W(x^{''})-F_W(x^{'})\Vert \ (\text{Lipschitz constant of $g$})\\&=\frac{1-L_g}{4}\Vert f_W(x^{''})+x^{''}-f_W(x^{'})-x^{'}\Vert\\&\geq \frac{1-L_g}{4}(\Vert x^{''}-x^{'}\Vert-\Vert f_W(x^{''})-f_W(x^{'})\Vert) \ \text{(triangle inequality)}\\&\geq \frac{(1-L_f)(1-L_g)}{4}\Vert x^{''}-x^{'}\Vert \ (\text{Lipschitz constant of $f_w$}).
\end{align*}
Let $y^{'}=P(x^{'})$, $y^{''}=P(x^{''})$, so $x^{'}=P^{-1}(y^{'})$, $x^{''}=P^{-1}(y^{''})$. This leads to 
$\Vert P^{-1}(y^{''})-P^{-1}(y^{'})\Vert\leq\frac{4}{(1-L_f)(1-L_g)}\Vert y^{''}-y^{'}\Vert.
$
This suggests that $L_{P^{-1}}\leq \frac{4}{(1-L_f)(1-L_g)}$, and it concludes the proof.
\end{proof}
\section{The Proof of Lemma \ref{lemma: preliminary}}
\label{sec:proof preliminary}
\begin{proof}
The optimality condition of $x^{k+1}$ leads to
$
    \nabla J^k(x^{k+1})+A^T\lambda^k+\rho^kA^T(Ax^k-b)+\alpha^k(x^{k+1}-x^k)=0.
    $
    We plug in $\lambda^{k+1}=\lambda^{k}+\rho^k(Ax^{k+1}-b)$  and arrange to obtain
    \begin{align}
     \nabla J^k(x^{k+1})+A^T\lambda^{k+1}-\rho^{k}A^TA(x^{k+1}-x^{k})+\alpha^k(x^{k+1}-x^k)=0. \label{eq: x optimality}  
    \end{align}
    That is 
    $\nabla J^k(x^{k+1})=-A^T\lambda^{k+1}+\rho^{k}A^TA(x^{k+1}-x^{k})-\alpha^k(x^{k+1}-x^k).  
    $
     Also the optimality condition of $x^k_*$ results in $\nabla J^k(x^k_*)+A^T\lambda^k_*=0$, due to the convexity of $J^k(x)$, we have
     \begin{align}
         &J^k(x^{k+1})\geq J^k(x^k_*)+(-A^T\lambda^k_*)^T(x^{k+1}-x^k_*).\label{ineq: convex 1}\\
         &J^k(x^k_*)\geq J^k(x^{k+1})+(-A^T\lambda^{k+1}+\rho^{k}A^TA(x^{k+1}-x^{k})-\alpha^k(x^{k+1}-x^k) )^T(x^k_*-x^{k+1}).\label{ineq: convex 2}
     \end{align}
We sum Inequalities \eqref{ineq: convex 1} and \eqref{ineq: convex 2} to obtain
$
    (A^T\lambda^{k+1}-A^T\lambda^k_*-\rho^{k}A^TA(x^{k+1}-x^{k})+\alpha^k(x^{k+1}-x^k))^T(x^k_*-x^{k+1})\geq 0.
$
 After rearranging terms, we have
 \begin{align*}
     &(\lambda^{k+1}-\lambda^k_*)^TA(x^k_*-x^{k+1})+\alpha^k(x^k-x^{k+1})^T(x^{k+1}-x^k_*)\geq \rho^k (x^{k+1}-x^k)^TA^TA (x^k_*-x^{k+1}).
 \end{align*}
Using the facts that $Ax^{k+1}-Ax^k_*=(Ax^{k+1}-b)-(Ax^k_*-b)=Ax^{k+1}-b=\frac{1}{\rho^k}(\lambda^{k+1}-\lambda^{k})$, and $a^Tb=ab^T$($a$ and $b$ are two vectors), we have
\begin{align*}
\\&
   \frac{1}{\rho^k}(\lambda^{k}-\lambda^{k+1})^T(\lambda^{k+1}-\lambda^k_*) +\alpha^k(x^k-x^{k+1})^T(x^{k+1}-x^k_*)\geq (x^{k+1}-x^k)^TA^T(\lambda^k-\lambda^{k+1}).
\end{align*}
\end{proof}
\section{Proof of Theorem \ref{theo: decrease u}}
\label{sec:proof decrease u}
\begin{proof}
We denote $u^k=(\lambda^k,x^k)$, $u^{k+1}=(\lambda^{k+1},x^{k+1})$, and $u^k_*=(\lambda^k_*,x^k_*)$.
Using the notation defined in Equation \eqref{eq: notation}, we have
\begin{align*}
    \langle u^k-u^{k+1},u^{k+1}-u^k_*\rangle_{M^k}\geq (x^{k+1}-x^k)^TA^T(\lambda^k-\lambda^{k+1}).
\end{align*}
Because $u^{k+1}-u^k_{*}=u^{k+1}-u^k+u^k-u^{k}_*$, it follows that
\begin{align}
  \langle u^k-u^{k+1},u^{k}-u^k_*\rangle_{M^k}\geq \Vert u^{k+1}-u^{k}\Vert^2_{M^k}+ (x^{k+1}-x^k)^TA^T(\lambda^k-\lambda^{k+1}). \label{ineq: convex 3}  
\end{align}
It holds that
\begin{align*}
    &\Vert u^k-u^k_*\Vert^2_{M^k}-\Vert u^{k+1}-u^k_*\Vert^2_{M^k}\\&=2\langle u^k-u^{k+1},u^k-u^k_*\rangle_{M^k}-\Vert u^{k+1}-u^k\Vert^2_{M^k}\\&\geq \Vert u^{k+1}-u^{k}\Vert^2_{M^k}+ 2(x^{k+1}-x^k)^TA^T(\lambda^k-\lambda^{k+1}) (\text{Inequality \eqref{ineq: convex 3}})\\&=\frac{1}{\rho^k} \Vert x^{k+1}-x^k\Vert^2_2+\alpha^k \Vert \lambda^{k+1}-\lambda^{k}\Vert^2_2+2(x^{k+1}-x^k)^TA^T(\lambda^k-\lambda^{k+1}) \ (\text{the definition of $\Vert\bullet\Vert^2_{M^k}$})\\&=\frac{1}{\rho^k} \Vert x^{k+1}-x^k\Vert^2_2+\alpha^k \Vert \lambda^{k+1}-\lambda^{k}\Vert^2_2-2(A(x^{k+1}-x^k))^T(\lambda^{k+1}-\lambda^{k}).
\end{align*}
Because $\alpha^k-\rho^kr(A^TA)>0$, it holds that $\delta^k=\frac{\alpha^k-\rho^kr(A^TA)}{2\rho^kr(A^TA)}> 0$. Let $\eta^k=\frac{\alpha^k}{(1+\delta^k)}>0$
\begin{align*}
    &\frac{1}{\rho^k} \Vert x^{k+1}-x^k\Vert^2_2+\alpha^k \Vert \lambda^{k+1}-\lambda^{k}\Vert^2_2-2(A(x^{k+1}-x^k))^T(\lambda^{k+1}-\lambda^{k})\\&\geq \frac{1}{\rho^k} \Vert x^{k+1}-x^k\Vert^2_2+\alpha^k \Vert \lambda^{k+1}-\lambda^{k}\Vert^2_2-\frac{1}{\eta^k}\Vert A(x^{k+1}-x^k)\Vert^2_2-\eta^k\Vert \lambda^{k+1}-\lambda^{k}\Vert^2_2\\&(\text{$2(A(x^{k+1}-x^k))^T(\lambda^{k+1}-\lambda^{k})\leq \frac{1}{\eta^k}\Vert A(x^{k+1}-x^k)\Vert^2_2+\eta^k\Vert \lambda^{k+1}-\lambda^{k}\Vert^2_2$})\\&\geq \frac{1}{\rho^k} \Vert x^{k+1}-x^k\Vert^2_2+\alpha^k \Vert \lambda^{k+1}-\lambda^{k}\Vert^2_2-\frac{r(A^TA)}{\eta^k}\Vert x^{k+1}-x^k\Vert^2_2-\eta^k\Vert \lambda^{k+1}-\lambda^{k}\Vert^2_2\\&(\Vert A(x^{k+1}-x^k)\Vert^2_2\leq r(A^TA)\Vert x^{k+1}-x^k\Vert^2_2  \ \text{where $r(A^TA)$ is the spectral radius of $A^TA$})\\&=(\frac{1}{\rho^k}-\frac{r(A^TA)}{\eta^k})\Vert x^{k+1}-x^k\Vert^2_2+(\alpha^k-\eta^k)\Vert \lambda^{k+1}-\lambda^k\Vert^2_2\\&=\frac{\delta^k r(A^TA)}{\alpha^k}\Vert x^{k+1}-x^k\Vert^2_2+\frac{\alpha^k\delta^k}{1+\delta^k}\Vert \lambda^{k+1}-\lambda^k\Vert^2_2\\&\geq \mu^k\Vert u^{k+1}-u^k\Vert^2_{M^k}.
\end{align*}
where $\mu^k=\min(\frac{\rho^k\delta^k r(A^TA)}{\alpha^k},\frac{\delta^k}{1+\delta^k})$
Because $(C^k,\rho^k,\tau^k,\alpha^k)$ guarantees that $\Vert u^{k+1}-u^{k+1}_{*}\Vert^2_{M^{k+1}}\leq \Vert u^{k+1}-u^{k}_{*}\Vert^2_{M^{k}}$, then we have
\begin{align}
    \Vert u^k-u^k_*\Vert^2_{M^k}\geq\Vert u^{k+1}-u^k_*\Vert^2_{M^k}+ \mu^k\Vert u^{k+1}-u^k\Vert^2_{M^k}\geq \Vert u^{k+1}-u^{k+1}_*\Vert^2_{M^{k+1}}+ \mu^k\Vert u^{k+1}-u^k\Vert^2_{M^k}. \label{ineq: constractive}
\end{align}
(a). We sum Inequality \eqref{ineq: constractive} from $k=0$ to $k=K$ to obtain
\begin{align*}
 \sum\nolimits_{k=0}^K \mu^k\Vert u^{k+1}-u^k\Vert^2_{M^k}\leq \Vert u^0-u^0_*\Vert^2_{M^0}. 
\end{align*}
Let $K\rightarrow \infty$, we have $\lim\nolimits_{k\rightarrow\infty}\mu^k\Vert u^{k+1}-u^k\Vert^2_{M^k}=0$. Because $\mu^k>0$, we have $\lim\nolimits_{k\rightarrow\infty}\Vert u^{k+1}-u^k\Vert^2_{M^k}=0$.
\\(b). From Inequality \eqref{ineq: constractive}, $\Vert u^k-u^k_*\Vert^2_{M^k}$ is nonincreasing, and $\Vert u^k-u^k_*\Vert^2_{M^k}>0$ has a lower bound. Therefore, $\Vert u^k-u^k_*\Vert^2_{M^k}$ is convergent.\\
(c). From (a) we know that $\lim\nolimits_{k\rightarrow\infty}\Vert  u^{k+1}-u^k\Vert^2_{M^k}=0$. That is $\lim\nolimits_{k\rightarrow\infty}\frac{1}{\rho^k}\Vert x^{k+1}-x^k\Vert^2_2+\alpha^k\Vert \lambda^{k+1}-\lambda^{k}\Vert^2_2=0$. Because $\alpha^k\geq D_1>0$ and $\rho^k\leq D_4<\infty$, we have $\lim\nolimits_{k\rightarrow\infty} x^{k+1}-x^k=0$ and $\lim\nolimits_{k\rightarrow\infty} \lambda^{k+1}-\lambda^k=0$. Because  $\lambda^{k+1}=\lambda^{k}+\rho^k(Ax^{k+1}-b)$ and $\rho^k\geq D_3>0$, then $\lim\nolimits_{k\rightarrow\infty} Ax^{k+1}-b=0$.\\ We take the limit on both sides of Equation \eqref{eq: x optimality} to obtain
\begin{align*}
 \lim\nolimits_{k\rightarrow\infty} \nabla J^k(x^{k+1})+A^T\lambda^{k+1}-\rho^{k}A^TA(x^{k+1}-x^{k})+\alpha^k(x^{k+1}-x^k)=0.   
\end{align*}
 Because $\rho^k$ and $\alpha^k$ are bounded, and $\lim\nolimits_{k\rightarrow\infty} x^{k+1}-x^k=0$, we have $
 \lim\nolimits_{k\rightarrow\infty} \nabla J^k(x^{k+1})+A^T\lambda^{k+1}=0   
$. In summary, we  prove that $u^{k+1}$ is a feasible solution to Equation \eqref{eq:layer problem}.
\end{proof}
\section{Descriptions of All Datasets}
\label{sec:dataset}
\indent All datasets are outlined below:\\
\indent  1. Karate \cite{lusseau2003bottlenose}. Karate contains the social ties among the members of a university karate club.\\
\indent 2. Dolphins \cite{lusseau2003bottlenose}. Dolphins is a social network of bottlenose dolphins, where edges represent frequent associations between dolphins.\\
\indent 3. Jazz \cite{gleiser2003community}. Jazz is a collaboration network between Jazz musicians. Each edge represents that two musicians have played together in a band.\\
\indent 4. Network Science \cite{newman2006finding}. Network Science is a coauthorship network of scientists working on network theory and experiment. Each edge represents two scientists who have co-authored a paper.\\
\indent 5. Cora-ML \cite{mccallum2000automating}. Cora-ML is a portal network of computer science research papers crawled by machine learning techniques.\\
\indent 6. Power Grid \cite{watts1998collective}. Power Grid is a topology network of the Western States Power Grid of the United States.\\
\indent 7. Memetracker \cite{leskovec2009meme}. The Memetracker keeps track of frequently used phrases on news social media.\\
\indent 8. Digg \cite{hogg2012social}. Digg is a reply network of the social news.\\
\indent 9. Deezer \cite{rozemberczki2019gemsec}. Deezer is an online music streaming service.  We used all nodes from Hungary. \\

\end{document}